\documentclass{IEEEtran}

\usepackage{amsmath}
\usepackage{amssymb}
\usepackage{amsfonts}
\usepackage{mathtools}
\usepackage{mathrsfs}
\usepackage{graphicx}
\usepackage{textcomp}
\usepackage{xcolor}
\usepackage{booktabs}
\usepackage{bm}

\usepackage{amsthm}
\newtheorem{theorem}{Theorem}

\newtheorem{definition}{Definition}
\newtheorem{assumption}{Assumption}
\newtheorem{proposition}{Proposition}
\newtheorem{corollary}{Corollary}

\usepackage[colorlinks=true,
            linkcolor=cyan,
            urlcolor=cyan,
            citecolor=cyan,
            anchorcolor=cyan]{hyperref}

\begin{document}

\title{
Performance Limits and Tradeoffs of Power Systems Synchronization Recovery in Complex-Frequency Representation
}

\author{Xiaoyu Peng, Zelin Sun, Zhongze Li, Feng Liu
\thanks{(\textit{Corresponding author: Feng Liu})}
\thanks{Xiaoyu Peng and Zhongze Li are with the Department of Electrical Engineering, Tsinghua University, Beijing, 100084, China.  (email: pengxy19@tsinghua.org.cn; lizz24@mails.tsinghua.edu.cn).}
\thanks{Zelin Sun is with the Department of Automation, Tsinghua University, Beijing, 100084, China.  (email: szl24@mails.tsinghua.edu.cn).}
\thanks{Feng Liu is with the Department of Electrical Engineering, Tsinghua University, Beijing, 100084, China, and also with the State Key Laboratory of Power System Operation and Control, Tsinghua University, 100084, China. (email: lfeng@tsinghua.edu.cn)}
}
\maketitle

\providecommand{\rev}[1]{\textcolor{blue}{#1}}

\begin{abstract}
Power-system synchronization studies mainly ask \textit{whether} synchronization can be achieved, whereas a complementary question is \textit{how well} post-disturbance synchronization recovery can be shaped. In inverter-dominated systems, this recovery involves coupled phase-angle and voltage-magnitude dynamics, for which complex frequency provides a unified representation. This paper establishes a conservation law for synchronization recovery in this representation. For a given system and its initial state, the logarithmic integral of the proposed recovery residual is fixed by the initial contraction rate and the right-half-plane (RHP) zero--pole structure. Hence, synchronization recovery performance cannot be improved uniformly over frequencies and variables unless the budget set by these terms changes; controller tuning that preserves them can only redistribute the residual response across frequencies. The conservation law applies to both linear and nonlinear recovery responses satisfying the mild analyticity conditions. It further yields finite-frequency peak--bandwidth limits and a lower bound on the time-domain integral of squared complex frequency recovery. A system-level extension recovers classical Bode-type sensitivity conservation as an LTI special case. Studies on a heterogeneous IEEE 39-bus system verify the theory. 
\end{abstract}


\section{Introduction}\label{sec: introduction}

\subsection{Motivation and Background}

{The growing penetration of inverter-based resources has strengthened the coupling between phase angle and voltage dynamics \cite{Peng_Impact_2025}. Due to their interactions, disturbances may excite synchronization responses that phase angle or frequency alone cannot fully describe \cite{Gutierrez-Florensa_Quasi_2026}. Accordingly, post-disturbance synchronization recovery must account for the coupled evolution of phase angle and voltage magnitude. Complex frequency \cite{Milano_Complex_2022}, which combines the normalized voltage rate and angular frequency, provides a unified representation of this coupled recovery. Existing complex-frequency studies mainly characterize synchronization conditions, propagation, and control, whereas convergence alone does not determine synchronization recovery performance. For example, stronger damping may suppress oscillatory excursions while leaving a slower residual recovery; more generally, improvement over one frequency range may be accompanied by degradation elsewhere \cite{Peng_Stability_2026}. This motivates the central question:}

{\emph{\textbf{For a given system and initial state, how well can its synchronization recovery performance be shaped?}}}

{Answering this question requires distinguishing whether controls change the overall budget on synchronization recovery performance, or merely redistribute it among frequencies or variables. This paper distinguishes them by establishing a conservation law for the complex-angle recovery determined by the system and its initial state.}

\begin{figure*}[htpb]
\centering
\includegraphics[width=0.9\textwidth]{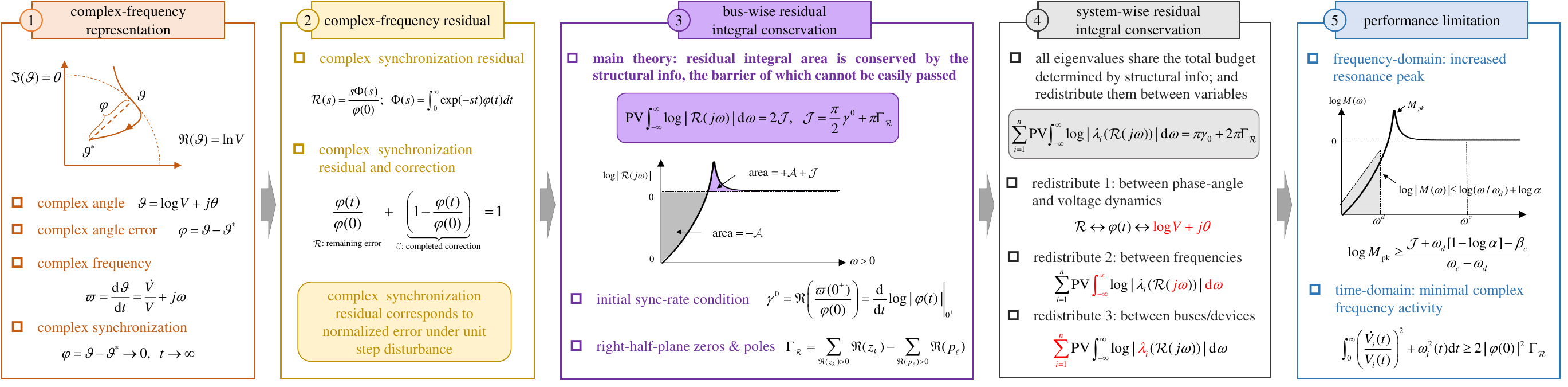}
\caption{{Overview of the proposed framework and its power-system implications. The system model and initial state set the synchronization recovery response and its residual. Then, the logarithmic recovery budget is determined by the initial contraction rate and the RHP-zero burden from dynamic cancellation. The conservation law yields finite-frequency peak and bandwidth limits and a lower bound on time-domain complex-frequency activity. Different initial states and observation buses generally have different budgets, but each obeys the same conservation law.}}
\label{fig: framework overview}
\end{figure*}

\subsection{Literature Review}

{Answering the above question requires two ingredients: a representation of coupled angle--voltage synchronization recovery and a theory that can distinguish a change in the overall recovery limitation from redistribution of recovery performance. The most relevant literature follows two lines, the complex-frequency representations of coupled synchronization and Bode-type fundamental performance limits.}

{The first line of the complex-frequency representation starts from the spatial variation of frequency during power-system transients. The frequency-divider formulation links local bus frequencies to network coupling and device dynamics \cite{Milano_Frequency_2017}, motivating the bus-level observation adopted here. Complex frequency extends this view by combining angular frequency and normalized voltage rate and linking them to complex-power and current dynamics \cite{Milano_Complex_2022}. Later studies clarified its geometry \cite{Milano_Geometrical_2022, Milano_Equivalence_2024, Sofos_Complex_2026} and developed a complex-frequency divider for network propagation \cite{Ponce_Complex_2024}. Others distinguished bus and internal-device complex frequencies \cite{Moutevelis_Taxonomy_2024} and gave modal interpretations of local responses \cite{Moutevelis_Modal_2024}. These works establish complex frequency as a meaningful representation of coupled angle--voltage dynamics. Their focus, however, is on how complex frequency is generated, propagated, or decomposed rather than on fundamental limits of synchronization recovery performance.}

{Related studies address coupled phase--voltage synchronization and control. Oscillator-based controls provide stability guarantees for the angle--voltage joint evolution of dispatchable virtual oscillator controls (dVOC) \cite{Colombino_Global_2019}. In the complex-frequency formulation, dVOC synchronization can be interpreted through complex droop \cite{He_Complex_2025}. Local synchronization \cite{Ponce_Local_2025} and augmented synchronization \cite{Yang_Augmented_2023} extend this analysis to device terminals and heterogeneous bus dynamics. These ideas have also been used to improve inverter transients \cite{DomingoEnrich_Dynamic_2024, Bernal_Complex_2025}. These works provide definitions, stability conditions, metrics, and control methods. They mainly address whether synchronization is achieved and how to shape a desired response, rather than whether synchronization recovery performance can be improved uniformly over all frequencies.}

The second line concerns fundamental performance limits. Classical Bode-type sensitivity integrals show that, under suitable pole-zero conditions, better disturbance rejection in one frequency range must be compensated elsewhere \cite{Seron_Fundamental_1997}. Our previous work \cite{Peng_Stability_2026} first introduced this tool into power systems and established sensitivity conservation and breaking laws. It showed that angle-voltage coupling redistributes a system-level disturbance rejection budget across physical variables, devices, and frequencies. Improvement in one component may therefore degrade another. Ref.~\cite{Hong_Waterbed_2026} connected this waterbed mechanism to inverter synchronization and explained the trade-off between response speed and stability margin. The limitation can be relaxed only by changing the loop structure or pole-zero conditions. 

{These studies share the same performance-limitation viewpoint but use different analytical objects. Existing Bode-type studies consider fixed feedback interconnections and prescribed disturbance-rejection channels. The post-disturbance synchronization recovery considered here, by contrast, is a free response determined by the system dynamics and the disturbance-induced initial state. Thus, the object to be constrained is not a prescribed disturbance-rejection transfer function, but the remaining synchronization error itself. Existing Bode-type conservation results are also mainly formulated for linear systems. Meanwhile, complex-frequency studies provide the representation and synchronization mechanisms, but not an exact conservation law for this recovery performance. This paper connects the two lines by constructing a synchronization residual directly from the recovery response and relating its integral to the initial contraction tendency and non-minimum-phase dynamics. Under linearization and additional assumptions, the result reduces to \cite{Peng_Stability_2026}.}

\subsection{Main Contribution}

 {Building on complex-frequency and related phase--magnitude synchronization theories, this work studies the structural synchronization recovery performance following a disturbance. The main contributions are summarized as follows.}

   
    {1) \textbf{A performance formulation is constructed for characterizing complex synchronization recovery.} Building on existing complex-frequency theory, the proposed residual characterizes the remaining coupled angle--voltage synchronization error after disturbances. Unlike existing analyses centered on stability, modal behavior, or controller-specific performance metrics, this formulation enables direct analysis of synchronization recovery performance and its fundamental limits.}

    {2) \textbf{A conservation law is established that reveals a fundamental limit of synchronization recovery performance.} Unlike classical Bode-type results formulated for predefined feedback input--output channels, the proposed conservation law applies directly to disturbance-induced synchronization recovery. It shows that recovery performance obeys a logarithmic budget determined by the initial contraction tendency and the non-minimum-phase structure. This converts otherwise case-dependent transient behavior into a quantitative law. Computation methods for both linear and nonlinear systems are derived.}

    {3) \textbf{Quantitative performance tradeoffs are exposed with control significance.} Existing controller-tuning studies typically observe compromises among damping, response speed, and transient excursions without distinguishing their underlying causes. The proposed theory separates changes in the recovery budget from redistribution of recovery performance within a fixed budget, and derives corresponding peak--bandwidth and complex-frequency-activity limits. This distinction clarifies when control changes the underlying recovery limitation and when it only relocates performance degradation, providing a principled basis for synchronization-oriented controller design.}

\subsection{Paper Organization}

Section~\ref{sec: complex frequency synchronization} defines complex-angle synchronization and the residual. Section~\ref{sec: residual integral} establishes the Bode-type conservation law. Section~\ref{sec: pole zero interpretation} interprets residual zeros and poles and gives computations for linear and nonlinear systems. Section~\ref{sec: limitations} derives frequency- and time-domain limits. Section~\ref{sec: extension to system level} extends the law to multiple buses and recovers earlier Bode-type sensitivity conservation as a special case. Section~\ref{sec: ieee39 validation} presents numerical validation, and Section~\ref{sec: conclusion} concludes the paper. Fig.~\ref{fig: framework overview} summarizes the framework. 

\textbf{Notation}: The real and complex number sets are $\mathbb R$ and $\mathbb C$, with $\mathbb C_+=\{s\in\mathbb C:\Re(s)>0\}$. The real and imaginary parts of $c\in\mathbb C$ are $\Re(c)$ and $\Im(c)$. LHP and RHP denote the left- and right-half planes, and $j$ is the imaginary unit. Symmetric Cauchy principal-value integration is denoted by $\operatorname{PV}\int$. The $L^p$ space contains functions with finite norm $\Vert f\Vert_p=(\int |f(x)|^p\mathrm dx)^{1/p}$.
The $C^p$ space contains functions with $p$ continuous derivatives.
The eigenvalue of a complex matrix $M$ is denoted by $\lambda(M)$.

\section{Power-System Synchronization in Complex Frequency Representation}\label{sec: complex frequency synchronization}

\subsection{General Power-System Model}

Consider a general power-system model described by the differential-algebraic equations
\begin{equation}\label{equ: dae model}
    \dot x=f(x,y),\qquad 0=g(x,y),
\end{equation}
where $x$ and $y$ collect the dynamic states and algebraic variables. The model includes the network and heterogeneous devices in inverter-dominated systems. The bus voltage is $v_i=h_i(x,y)\in\mathbb C$. In the synchronous frame, write $v_i(t)=V_i(t)\exp(j\theta_i(t))$, where $V_i(t)>0$ and $\theta_i(t)$ is the phase angle relative to the frame rotating at $\omega_s=100\pi$ rad/s. The angular-frequency deviation is $\omega_i(t)=\dot\theta_i(t)$.

\textbf{Applicability}:
{The analysis is not tied to a specific electromechanical, electromagnetic, linear, or nonlinear model. Under the well-posedness condition below, the system and initial state uniquely determine the synchronization response and, therefore, the recovery budget introduced below. The result is therefore universal in form.}

\begin{assumption}[Well-Posed System Model]\label{assump: trajectory}
For the considered initial state $(x_0,y_0)$ satisfying $g(x_0,y_0)=0$, \eqref{equ: dae model} admits a unique non-singular forward solution.
\end{assumption}

\subsection{Complex Angle, Complex Frequency, and Synchronization}

\begin{definition}[Complex Angle and Complex Frequency \cite{Milano_Complex_2022}]\label{def: complex angle frequency}
The complex angle $\vartheta_i$ and complex frequency $\varpi_i$ at bus $i$ are respectively defined as
\begin{equation}\label{equ: complex frequency}
\begin{aligned}
    \vartheta_i(t)&=\log V_i(t)+j\theta_i(t),\\
    \varpi_i(t)&=\dot\vartheta_i(t)=\frac{\dot V_i(t)}{V_i(t)}+j\omega_i(t)=\frac{\dot v_i(t)}{v_i(t)}.
\end{aligned}
\end{equation}
\end{definition}

Here, $\Re(\varpi_i)$ is the normalized voltage rate and $\Im(\varpi_i)$ is the angular-frequency deviation from $\omega_s$. Complex synchronization requires both phase angle and voltage to converge to their post-disturbance equilibrium.

\begin{definition}[Complex Synchronization]\label{def: synchronization}
The $i$-th bus achieves synchronization if there exists a constant $\vartheta_i^*\in\mathbb C$ such that
\begin{equation}\label{equ: complex angle synchronization}
    \vartheta_i(t)\rightarrow\vartheta_i^*,
    \qquad t\rightarrow\infty.
\end{equation}
Define the complex-angle synchronization error as
\begin{equation}\label{equ: synchronization error}
{\varphi_i(t)=\vartheta_i(t)-\vartheta_i^*, \quad \dot\varphi_i(t)=\varpi_i(t).}
\end{equation}
Equivalently, complex synchronization requires $\varphi_i\rightarrow0$.
\end{definition}

\subsection{Complex Synchronization Residual}

Set the fault-clearing time to $t=0$, with $\varphi(0)\neq0$. Under Assumption~\ref{assump: trajectory}, the system model, initial state, and selected bus uniquely determine $\varphi_i(t)$. We construct an analytic object from this response and omit the bus subscript $i$ for brevity.

\begin{assumption}[Synchronized Response]\label{assump: integrable trajectory}
The complex synchronization error satisfies $\varphi\in L^1\cap C^1$, $\varphi(t)\rightarrow0$, $\varpi=\dot\varphi\in L^1$, and $\varphi(0)\neq0$.
\end{assumption}

\begin{definition}[Complex Synchronization Residual]\label{def: residual function}
The Laplace transform of complex-angle synchronization error is
\begin{equation}\label{equ: Phi definition}
\boxed{\Phi(s)=\int_0^\infty \exp(-st)\cdot\varphi(t){\rm d}t.}
\end{equation}
The corresponding complex synchronization residual is
\begin{equation}\label{equ: residual definition}
    \boxed{\mathcal R(s)=\frac{s\Phi(s)}{\varphi(0)}.}
\end{equation}
\end{definition}
The system model, initial state, and observation bus uniquely determine $\mathcal R(s)$. It is therefore a system- and initial-condition-dependent quantity, not an arbitrary response descriptor. Its unit-step response also has a direct physical interpretation as follows.
\begin{proposition}[Residual-Correction Representation]\label{prop: residual frequency representation}
Under Assumption~\ref{assump: integrable trajectory}, define the transform of the complex frequency as
\begin{equation}\label{equ: Omega definition}
    \Omega(s)=\int_0^\infty \exp(-st)\cdot\varpi(t){\rm d}t,
\end{equation}
and the normalized correction function
\begin{equation}\label{equ: correction function}
    \mathcal C(s)=-\frac{\Omega(s)}{\varphi(0)}=1-\mathcal R(s).
\end{equation}
Their unit-step responses satisfy
\begin{equation}\label{equ: residual and correction step response}
    \mathcal L^{-1}\!\left\{\frac{\mathcal R(s)}{s}\right\}
    =\frac{\varphi(t)}{\varphi(0)}, \quad 
    \mathcal L^{-1}\!\left\{\frac{\mathcal C(s)}{s}\right\}
    =1-\frac{\varphi(t)}{\varphi(0)}.
\end{equation}
Also, $\mathcal R(0)=0$, $\mathcal C(0)=1$, $\mathcal R(s)\to1$, and $\mathcal C(s)\to0$ as $|s|\to\infty$ in the closed RHP.
\end{proposition}

\begin{proof}
The proof is provided in Appendix~\ref{app:proof:residual-correction}.
\end{proof}

Prop.~\ref{prop: residual frequency representation} expresses the recovery through complementary operators. The $\mathcal R$ is the residual fraction of the initial synchronization error, whereas $\mathcal C$ is the corrected fraction, analogous to the relation between sensitivity and complementary sensitivity functions in classical control theory \cite{Seron_Fundamental_1997}.
In other words, $\mathcal R$ indicates the remaining error of the current state with respect to the steady state, per-unit normalized with the initial state.
\begin{equation}\label{equ: residual correction balance}
    \underbrace{\frac{\varphi(t)}{\varphi(0)}}_{\mathcal R\text{: remaining error}}
    +\underbrace{\left(1-\frac{\varphi(t)}{\varphi(0)}\right)}_{\mathcal C\text{: completed correction}}
    =1.
\end{equation}

\section{Integral Conservation for Complex Synchronization}\label{sec: residual integral}

This section establishes the main conservation law for synchronized dynamics satisfying the stated conditions.

\begin{assumption}[Residual Analyticity]\label{assump: residual analyticity}
The residual $\mathcal R(s)$ is meromorphic in the open RHP and admits continuous boundary values on the imaginary axis, with no imaginary-axis poles and zeros except a finite-order structural zero at $s=0$. 
Moreover, $\varphi\in C^2$ and $\dot\varpi\in L^1$.
\end{assumption}

Intuitively, Assumption~\ref{assump: residual analyticity} excludes an accumulation of unstable zeros or singularities in the synchronization residual. 
It is satisfied if there exist only finitely many RHP zeros and poles. For a stable synchronized response satisfying Assumption~\ref{assump: integrable trajectory}, it indicates that $\mathcal R$ is analytic in the open RHP.

Assumptions~\ref{assump: integrable trajectory} and \ref{assump: residual analyticity} ensure the expansion 
\begin{equation}\label{equ: analytic asymptote}
    \mathcal R(s)=1+\frac{a_1}{s}+o(|s|^{-1}),
\end{equation}
as $|s|\to\infty$ with $\Re(s)\geq0$.
with a remainder uniform on RHP semicircles:
\begin{equation}\label{equ: uniform analytic asymptote}
    \lim_{R\to\infty}
    \sup_{\substack{\Re(s)\geq0\\|s|=R}}
    R\left|\mathcal R(s)-1-\frac{a_1}{s}\right|=0.
\end{equation}

Along the positive real axis, \eqref{equ: analytic asymptote} gives $a_1=\lim_{\sigma\to\infty}\sigma[\mathcal R(\sigma)-1]$. The Laplace initial-value theorem then yields
\begin{equation}\label{equ: a1}
\begin{aligned}
    a_1
      =\lim_{\sigma\rightarrow\infty}
      \frac{\sigma\int_0^\infty \exp(-\sigma t)\cdot\varpi(t){\rm d}t}{\varphi(0)}=\frac{\varpi(0^+)}{\varphi(0)}.
\end{aligned}
\end{equation}

Let $z_i$ and $p_i$ denote the open-RHP zeros and poles of $\mathcal{R}(s)$, including algebraic multiplicities, and define
\begin{equation}\label{equ: gamma zp definition}
    \boxed{\Gamma_{\mathcal R}
    =\sum_{\Re(z_i)>0}\Re(z_i)
    -\sum_{\Re(p_i)>0}\Re(p_i).}
\end{equation}
Besides, define the initial relative contraction rate of the complex-angle synchronization error $\gamma_0$ as
\begin{equation}\label{equ: initial contraction rate}
\boxed{\gamma_0=\Re\left(\frac{\varpi(0^+)}{\varphi(0)}\right)=\left.\frac{{\rm d}}{{\rm d}t}\log|\varphi(t)|\right|_{0^+}.}
\end{equation}

Then, building on the above two definitons, denote the logarithmic residual budget as
\begin{equation}\label{equ: residual budget definition}
    \mathcal J=\frac{\pi}{2}\Re(a_1)+\pi\Gamma_{\mathcal R}.
\end{equation}

\begin{theorem}[Complex Synchronization Residual Integral]\label{theo: residual integral}
Every residual satisfying Assumption~\ref{assump: residual analyticity} obeys
\begin{equation}\label{equ: bilateral integral theorem}
\boxed{\operatorname{PV}\int_{-\infty}^{\infty}
\log|\mathcal R(j\omega)|{\rm d}\omega
=2\mathcal J.}
\end{equation}
If $\mathcal R$ is induced by a synchronized trajectory satisfying Assumption~\ref{assump: integrable trajectory}, then
\begin{equation}\label{equ: physical integral theorem}
    \boxed{\mathcal J=\frac{\pi}{2}\gamma_0+\pi\Gamma_{\mathcal R},}
\end{equation}
\end{theorem}

\begin{proof}
The proof is provided in Appendix~\ref{app:proof:residual-integral}.
\end{proof}

\begin{corollary}[Stable Synchronized Response]\label{cor: stable integral}
Suppose Assumptions~\ref{assump: integrable trajectory} and \ref{assump: residual analyticity} hold for the synchronized response. Then,
\begin{equation}\label{equ: stable residual integral}
    \Gamma_{\mathcal R}
    =\sum_{\Re(z_i)>0}\Re(z_i),\quad
    \mathcal J=\frac{\pi}{2}\gamma_0
    +\pi\sum_{\Re(z_i)>0}\Re(z_i).
\end{equation}
\end{corollary}

\begin{proof}
The proof is provided in Appendix~\ref{app:proof:stable-response}.
\end{proof}

\textbf{Applicability and power-system implication of Theorem~\ref{theo: residual integral}:}
Once the system models and initial state are specified, the result is independent of model form. 
The identity is universal in form, not in the value of $\mathcal J$. Different initial states generally give different budgets, but each system--initial combination has a unique one. 

The term $\gamma_0$ describes the initial contraction or expansion of the complex-angle error. The RHP-zero term measures non-minimum-phase cancellation among response components and network paths. 
For LTI models, they reduce to standard finite-dimensional pole-zero conditions.
Controller tuning can therefore reshape the response while preserving these terms, or change the budget by altering them. The former redistributes constrained performance across frequencies or physical variables; only the latter changes the fundamental burden.

\section{Interpretation and Computation of Residual Zeros and Poles}\label{sec: pole zero interpretation}

\subsection{Poles as Response Singularities}

Under Assumption~\ref{assump: residual analyticity}, a pole of a meromorphic continuation of $\Phi_i(s)$ or $\mathcal R_i(s)$ is a singularity of the selected response transform. If the inverse-Laplace contour crosses an isolated pole $p$, its residue contributes a term proportional to $\exp(pt)$. Here, $\Re(p)$ gives the growth or decay rate and $\Im(p)$ the oscillation frequency. 
For a stable synchronized response satisfying Assumption~\ref{assump: integrable trajectory}, $\Phi_i$ is analytic in $\mathbb C_+$ and has no RHP pole. In a more general meromorphic representation, an RHP pole produces exponential growth and is incompatible with stable synchronization.
For LTI systems, residual poles are state-matrix eigenmodes; see Section~\ref{subsec: lti characterization}. 

\subsection{Zeros as Exponentially Weighted Recovery Cancellation}

\begin{proposition}[Weighted-Cancellation Characterization]\label{prop: zero cancellation}
$z=\sigma+j\omega_z\in\mathbb C_+$ is a zero of $\mathcal R$ if and only if $\int_0^\infty \exp(-\sigma t)\cdot\exp(-j\omega_z t)\cdot\varphi(t){\rm d}t=0$.
\end{proposition}

\begin{proof}
The proof is provided in Appendix~\ref{app:proof:weighted-cancellation}.
\end{proof}

The result gives a physical meaning to an RHP residual zero. Within an exponentially weighted window of scale $1/\sigma$, the complex-angle error cancels exactly after demodulation at $\omega_z$. Thus, a zero is not an independent mode. It results from cancellation among response components and propagation paths seen through the selected channel. 
In power systems, such contributions may arise from coherent machine groups and inverter controls. Different disturbance locations or initial states can therefore produce different residual zeros.

\subsection{Characterization for Linear-Time-Invariant Systems}\label{subsec: lti characterization}
For LTI systems, residual zeros and poles follow directly from transfer-function expressions.

\subsubsection{Unspecified Disturbance Directions}
Consider the linearized model with stable open-loop state matrix $A$ as $\dot x=Ax$, $\varphi_i=C_ix$ and $x(0)=x_0$.

\begin{proposition}[Initial-Condition Transfer-Function Representation]\label{prop: lti initial condition representation}
For a fixed $x_0$ satisfying $\varphi_i(0)=C_ix_0\neq0$, the complex angle error and residual are
\begin{equation}\label{equ: residual lti realization}
    \Phi_i(s)=C_i(sI-A)^{-1}x_0,\quad
    \mathcal R_i(s)=1+\frac{C_iA(sI-A)^{-1}x_0}
    {C_ix_0}.
\end{equation}
Its poles are the modes that are both excited by $x_0$ and observed by $C_i$. Its nonzero zeros coincide with the transmission zeros of the initial-condition channel $C_i(sI-A)^{-1}x_0$, while $s=0$ is the structural residual zero.
\end{proposition}

\begin{proof}
The proof is provided in Appendix~\ref{app:proof:lti-initial-condition}.
\end{proof}

Suppose $A$ is diagonalizable with $A=\sum_m \lambda_m v_m w_m^H$, where $v_m$ and $w_m$ are biorthogonally normalized right and left eigenvectors. Then, $\Phi_i(s)=\sum_m{(C_iv_m)(w_m^Hx_0)}/{(s-\lambda_m)}$.
It links the poles of $\Phi_i$ to the eigenvalues of $A$ and makes two selection mechanisms explicit: $w_m^Hx_0$ quantifies disturbance excitation, while $C_iv_m$ quantifies bus observation. 

\subsubsection{Prescribed Disturbance Directions}

Consider fixed disturbance directions described by
\begin{equation}\label{equ: lti forced model}
    \dot x=Ax+B_d d,
    \qquad \varphi_i=C_ix+D_{id}d.
\end{equation}
For a predetermined scalar disturbance direction $q$, the disturbance-to-complex-angle-error transfer function is
\begin{equation}
    G_{i,q}(s)=[C_i(sI-A)^{-1}B_d+D_{id}]q.
\end{equation}

\begin{proposition}[Relation to Conventional Transfer-Function Zeros]\label{cor: conventional transfer zeros}
Suppose $D_{id}=0$ and an impulse disturbance $d(t)=q\delta(t)$ is applied to \eqref{equ: lti forced model}. If $C_iB_dq\neq0$, then $x_0=B_dq$ and
\begin{equation}\label{equ: residual from physical channel}
    \boxed{\Phi_i(s)=G_{i,q}(s),\quad
    \mathcal R_i(s)=\frac{sG_{i,q}(s)}{C_iB_dq}.}
\end{equation}
After minimal realization, the nonzero zeros and poles of $\mathcal R_i$ coincide with the conventional transmission zeros and poles of $G_{i,q}$, respectively, and $s=0$ is additionally introduced as the structural synchronization zero.
\end{proposition}

\begin{proof}
The proof is provided in Appendix~\ref{app:proof:transfer-zeros}.
\end{proof}

In power-system studies, $B_dq$ can represent an initial state jump or an impulse applied through specified controls or buses, while $C_i$ selects the observation bus. Once both are fixed, numerical software can directly give the residual poles and zeros from the transfer functions. 



\subsection{Computation from Nonlinear System Models}\label{subsec: trajectory based computation}

The budget in \eqref{equ: residual budget definition} requires two terms: the initial rate $\gamma_0$ and the zero-pole term $\Gamma_{\mathcal R}$. The initial-rate term follows directly from \eqref{equ: initial contraction rate} and the specified initial state and model. For LTI models, the RHP-zero term follows from the transfer-function representations above. For nonlinear models, $\Gamma_{\mathcal R}=\sum_{\Re(z_i)>0}\Re(z_i)$ can be estimated from simulated or measured response samples as follows.

\subsubsection{Contour-Moment Principle}
Let $\mathcal D\subset\mathbb C_+$ be a bounded domain whose positively oriented boundary contains no zero of $\Phi$. Define the zeroth- and first-order contour moments as
\begin{equation}\label{equ: residual zero moments}
    \nu_q(\mathcal D)=\frac{1}{2\pi j}
    \oint_{\partial\mathcal D}s^q\frac{\Phi'(s)}{\Phi(s)}{\rm d}s,
    \qquad q=0,1.
\end{equation}
By the argument principle \cite{Delves_Numerical_1967},
\begin{equation}\label{equ: residual zero moment identities}
    \nu_0(\mathcal D)=\sum_{z_i\in\mathcal D}m_i,
    \qquad
    \nu_1(\mathcal D)=\sum_{z_i\in\mathcal D}m_iz_i,
\end{equation}
where $m_i$ is the algebraic multiplicity. Thus, $\nu_0$ gives the enclosed zero count as an integer-valued convergence check, while $\nu_1$ gives the zero sum. If $\mathcal D$ encloses all RHP zeros, then
\begin{equation}\label{equ: gamma from first moment}
    \Gamma_{\mathcal R}=\Re\bigl(\nu_1(\mathcal D)\bigr)
    =\Re\left(\frac{1}{2\pi j}
    \oint_{\partial\mathcal D}s\frac{\Phi'(s)}{\Phi(s)}{\rm d}s\right).
\end{equation}

\subsubsection{Numerical Evaluation from Nonlinear Response Data}
Suppose the simulation provides $\varphi_n=\varphi(t_n)$ over $0=t_0<\cdots<t_N=T$. Let $\Delta t_n=t_{n+1}-t_n$, which gives
\begin{equation}\label{equ: discrete trajectory laplace transform}
\begin{aligned}
    \widehat\Phi_T
    =\sum_{n=0}^{N-1}
    e^{-st_n}\varphi_n{\Delta t_n},\;
    \widehat\Phi_T'
    =-\sum_{n=0}^{N-1}
    t_ne^{-st_n}\varphi_n{\Delta t_n}.
\end{aligned}
\end{equation}
Use nested rectangular search regions
\begin{equation}\label{equ: D_r search region}
    \mathcal D_r=\mathcal D_{\epsilon_r,S_r,\Omega_r}
=\{s:\epsilon_r\leq\Re(s)\leq S_r,\ |\Im(s)|\leq\Omega_r\}
\end{equation}
where $\epsilon_{r+1}<\epsilon_r$, $S_{r+1}>S_r$, and $\Omega_{r+1}>\Omega_r$. In practice, the outer search range should be chosen consistently with the dynamical bandwidth represented by the adopted model rather than extended arbitrarily. For example, an electromechanical / electromagnetic model typically represents dynamics only up to the order of $10^2$~Hz / $10^3$~Hz. Therefore, $S_r$ and $\Omega_r$ can first be selected to cover the corresponding model-resolvable decay/growth rates and oscillation frequencies, and then enlarged with a margin until the contour moments no longer change, providing evidence that all RHP zeros meaningful to the adopted model have been enclosed. Let $s_m^{(r)}$, $m=0,\ldots,M_r$, be counterclockwise boundary samples with $s_{M_r}^{(r)}=s_0^{(r)}$, and define
$\widehat g_q(s)=s^q\widehat\Phi_T'(s)/\widehat\Phi_T(s)$. For $q=0,1$, the two moments are approximated by
\begin{equation}\label{equ: discrete residual zero moments}
    \widehat\nu_q(\mathcal D_r)=\frac{1}{2\pi j}
    \sum_{m=0}^{M_r-1}
    \frac{\widehat g_q(s_m^{(r)})+\widehat g_q(s_{m+1}^{(r)})}{2}
    \bigl(s_{m+1}^{(r)}-s_m^{(r)}\bigr).
\end{equation}
Accept the zero count only when $\widehat\nu_0$ is close to a nonnegative integer with negligible imaginary part. Expand the contour and refine the boundary grid until both moments converge, e.g., $|\widehat\nu_q(\mathcal D_{r+1})-\widehat\nu_q(\mathcal D_r)|\leq\tau_q$ for $q=0,1$. Then evaluate $\widehat\Gamma_{\mathcal R}=\Re(\widehat\nu_1)$ and $\widehat{\mathcal J}=\gamma_0{\pi}/{2}+\pi\widehat\Gamma_{\mathcal R}$. This procedure estimates the nonlinear residual budget from simulation or measurement data without linearization, subject to the stated regularity and convergence checks.

\section{Frequency- and Time-Domain Recovery Performance Limitations}\label{sec: limitations}

\subsection{Frequency-Domain Limits: Peak--Bandwidth Trade-Off}\label{subsec: frequency limitations}

The contour-integral equality fixes the total logarithmic area but not its frequency distribution. A finite-frequency limit follows by combining the global budget with low-frequency attenuation and high-frequency roll-off requirements, as in standard Bode-integral analysis \cite{Seron_Fundamental_1997}.

For a synchronization residual, define the magnitude as
\begin{equation}\label{equ: paired residual definitions}
    M(\omega)=\sqrt{|\mathcal R(j\omega)\mathcal R(-j\omega)|},\quad \omega\geq0.
\end{equation}
Then, Theorem~\ref{theo: residual integral} yields
\begin{equation}\label{equ: positive frequency budget}
    \int_0^\infty\log M(\omega){\rm d}\omega
    =\frac{1}{2}\operatorname{PV}\int_{-\infty}^{\infty}
    \log|\mathcal R(j\omega)|{\rm d}\omega
    =\mathcal J.
\end{equation}

\begin{figure}[htbp]
\centering
\includegraphics[width=\linewidth]{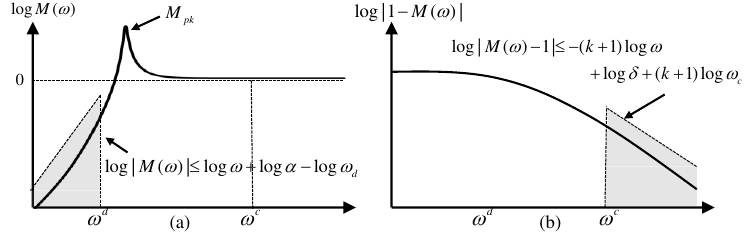}
\caption{Conditions and implication of Theorem~\ref{theo: finite frequency peak}. (a) The low-frequency envelope $\log M(\omega)\leq\log\alpha+\log(\omega/\omega_d)$ on $(0,\omega_d]$ forces a compensating peak in $(\omega_d,\omega_c)$ under a finite tail budget. (b) The tail $|M(\omega)-1|$ is bounded beyond the effective recovery bandwidth $\omega_c$.}
\label{fig: residual_frequency_constraints}
\end{figure}

\begin{theorem}[Finite-Frequency Residual Peak Bound]\label{theo: finite frequency peak}
Suppose Assumptions~\ref{assump: integrable trajectory} and \ref{assump: residual analyticity} hold. Let $0<\omega_d<\omega_c$, $k>0$, $0<\delta<1/2$, and $0<\alpha<1$. Suppose the following bounds hold:
\begin{itemize}
    \item $M(\omega)\leq\alpha(\omega/\omega_d)$ for all $\omega\in[0,\omega_d]$;
    \item $|M(\omega)-1|\leq\delta(\omega_c/\omega)^{1+k}$ for all $\omega\in[\omega_c,\infty)$.
\end{itemize}
Define $M_{\rm pk}=\sup_{\omega\in(\omega_d,\omega_c)}M(\omega)$. Then
\begin{equation}\label{equ: peak lower bound}
\boxed{\log M_{\rm pk}
\geq
\frac{\mathcal J+\omega_d[1-\log\alpha]-\beta_c}
{\omega_c-\omega_d},
\qquad
\beta_c=\frac{3\delta\omega_c}{2k}.}
\end{equation}
\end{theorem}


\begin{proof}
The proof is provided in Appendix~\ref{app:proof:finite-frequency-peak}.
\end{proof}

\textbf{Applicability and power system insights of Theorem~\ref{theo: finite frequency peak}:}
Fig.~\ref{fig: residual_frequency_constraints} illustrates the parameters in the the theorem. The first condition is equivalent to $\log M(\omega)\leq\log\alpha+\log(\omega/\omega_d)$ for $0<\omega\leq\omega_d$,
which gives a linear low-frequency envelope in $(\log\omega,\log M)$ coordinates. The second condition gives $\log|M-1|\leq(k+1)\log(\omega_c/\omega)+\log\delta$ for $\omega\geq\omega_c$. Here, $\omega_c$ is the highest frequency range with substantial recovery dynamics, while $\delta$ and $k$ set the tail magnitude and decay rate. Such a limit reflects finite device bandwidth, filtering, fast unmodeled dynamics, and noise \cite{Peng_Stability_2026}.

The low-frequency slope condition in Theorem~\ref{theo: finite frequency peak} creates a negative logarithmic area of at least $\omega_d[1-\log\alpha]$, while the high-frequency tail can supply at most $\beta_c$; see Appendix~\ref{app:proof:finite-frequency-peak}. The remaining area required by $\mathcal J$ must lie in $(\omega_d,\omega_c)$, giving \eqref{equ: peak lower bound}. In particular, $M_{\rm pk}>1$ is unavoidable if $\mathcal J+\omega_d[1-\log\alpha]>\beta_c$. A smaller $\alpha$, wider design band, larger RHP-zero burden, or narrower recovery bandwidth raises the unavoidable intermediate-frequency peak. Increasing $\omega_c$ can relax the trade-off only by extending recovery over a wider frequency range.

In power-system terms, $M_{\rm pk}$ measures the worst amplification of oscillatory synchronization. The parameter $\omega_c$ represents the usable device frequency range and acts as an effective recovery bandwidth for the specified residual. For the LTI channels in Section~\ref{subsec: lti characterization}, it reduces to the conventional closed-loop bandwidth.

\begin{corollary}[Complex Synchronization Residual Bandwidth]\label{cor: bandwidth lower bound}
Let $\omega_d=\kappa\omega_c$ with $0<\kappa<1$. Suppose the low-frequency condition of Theorem~\ref{theo: finite frequency peak} holds on $[0,\kappa\omega_c]$, the intermediate peak satisfies $M_{\rm pk}\leq\overline M$ with $\overline M>1$, and the high-frequency condition of the same theorem holds. Define
\begin{equation}\label{equ: bandwidth denominator}
    d_{\rm bw}=\kappa(\log\alpha-1)+(1-\kappa)\log\overline M
    +\frac{3\delta}{2k}.
\end{equation}
Suppose $\mathcal J>0$. If $d_{\rm bw}>0$, then necessarily
\begin{equation}\label{equ: bandwidth budget condition}
    \omega_c\geq{\mathcal J}/{d_{\rm bw}}.
\end{equation}
If $d_{\rm bw}\leq0$, these requirements are incompatible.
\end{corollary}

\begin{proof}
The proof is provided in Appendix~\ref{app:proof:bandwidth}.
\end{proof}

Corollary~\ref{cor: bandwidth lower bound} provides a feasibility test before detailed inverter tuning. Given low-frequency recovery, an allowable oscillatory peak, and high-frequency roll-off, it gives the minimum effective recovery bandwidth required by the budget. If $d_{\rm bw}<0$, the specifications are incompatible and must be relaxed or changed.

\subsection{Time-Domain Limits: Minimum Complex-Frequency Activity in Recovery}\label{subsec: time limitations}

The residual integral also gives a lower bound on total time-domain complex-frequency activity during synchronization.


\begin{theorem}[Minimum Complex-Frequency Activity]\label{theo: frequency activity lower bound}
Suppose Assumptions~\ref{assump: integrable trajectory} and \ref{assump: residual analyticity} hold and the error $\varpi(t)\in L^2$. Then
\begin{equation}\label{equ: general frequency activity lower bound}
\begin{aligned}
    \boxed{\int_0^\infty|\varpi(t)|^2{\rm d}t=\int_0^\infty\left(\frac{\dot V_i(t)}{V_i(t)}\right)^2+\omega_i^2(t){\rm d}t\geq2|\varphi(0)|^2\Gamma_{\mathcal R}.}
\end{aligned}
\end{equation}
\end{theorem}

\begin{proof}
The proof is provided in Appendix~\ref{app:proof:activity}.
\end{proof}

\textbf{Applicability and power-system implication of Theorem~\ref{theo: frequency activity lower bound}:}
The left-hand side of \eqref{equ: general frequency activity lower bound} combines voltage- and phase-angle-rate activity and therefore measures time-domain recovery effort. It also matches the metric used in \cite{Zhong_Improving_2022}, thereby providing a theoretical basis for that metric. 
On the right-hand side of \eqref{equ: general frequency activity lower bound}, the RHP-zero term represents a non-minimum-phase recovery path. Thus, a minimal amount of complex-frequency activity is unavoidable once the system and disturbances are determined. A controller may divide this activity between frequency and voltage or spread it over a longer interval, but cannot make the total arbitrarily small for fixed initial error and RHP-zero burden. Hence, the dynamic enhancement of angle/voltage is at the expense of worse voltage/angle transients elsewhere.

For a finite recovery horizon $T$, define the tail activity as $\epsilon_{\varpi}(T)=\int_T^\infty|\varpi(t)|^2{\rm d}t$. Then, Theorem~\ref{theo: frequency activity lower bound} and $\int_0^T|\varpi(t)|^2{\rm d}t\leq T\max_{0\leq t\leq T}|\varpi(t)|^2$ directly yield
\begin{equation}\label{equ: finite horizon peak bound}
    \max_{0\leq t\leq T}|\varpi(t)|
    \geq
    \sqrt{\frac{
    \max\{2|\varphi(0)|^2\Gamma_{\mathcal R}-\epsilon_{\varpi}(T),0\}}{T}}.
\end{equation}
Thus, completing most of the correction before $T$ requires a larger peak complex-frequency deviation as $T$ decreases.

\section{Extension to System-Level Integral Conservation}\label{sec: extension to system level}

\subsection{General System-Level Conservation}\label{subsec: general system level conservation}
The bus-level residual in Theorem~\ref{theo: residual integral} can be extended to a system-level form. Along the post-fault solution $(x(t),y(t))$ of \eqref{equ: dae model}, assume $g_y$ is nonsingular and let $\boldsymbol{\vartheta}=[\vartheta_1,\ldots,\vartheta_n]^T$. Eliminating algebraic variations gives
\begin{equation}\label{equ: nonlinear variational jacobians}
A(t)=f_x-f_y g_y^{-1}g_x,\qquad
C(t)=\boldsymbol{\vartheta}_x-\boldsymbol{\vartheta}_y g_y^{-1}g_x,
\end{equation}
where all Jacobians are evaluated along $(x(t),y(t))$. The state-flow Jacobian $\Psi(t)=\partial x(t)/\partial x_0$ satisfies
\begin{equation}\label{equ: nonlinear variational flow}
\dot\Psi(t)=A(t)\Psi(t),\qquad \Psi(0)=I.
\end{equation}

Consider an admissible $n$-dimensional perturbation $\eta$ of the fault-clearing state, $x_0(\eta)=x_0+B_0\eta+o(\|\eta\|)$, with $B_0\in\mathbb R^{n_x\times n}$. Let $\boldsymbol{\varphi}(t;\eta)=\boldsymbol{\vartheta}(t;\eta)-\boldsymbol{\vartheta}^{*}$ be the bus complex-angle error. Its differential system-level form is
\begin{equation}\label{equ: nonlinear differential synchronization error}
\bm{\varphi}(t):=\left.\frac{\partial\boldsymbol{\varphi}(t;\eta)}{\partial\eta}\right|_{\eta=0}
=C(t)\Psi(t)B_0,\;
\bm{\varphi}_0=\bm{\varphi}(0)=C_0B_0,
\end{equation}
where $\det\bm{\varphi}_0\neq0$. In direct analogy with \eqref{equ: Phi definition}-\eqref{equ: residual definition}, define the system-level complex residual
\begin{equation}\label{equ: nonlinear matrix residual definition}
\mathbf\Phi(s)=\int_0^\infty \exp(-st)\bm{\varphi}(t){\rm d}t,\quad
\boldsymbol{\mathcal R}(s)=s\mathbf\Phi(s)\bm{\varphi}_0^{-1}.
\end{equation}
Thus, $\mathcal L^{-1}\{\boldsymbol{\mathcal R}(s)/s\}=\bm{\varphi}(t)\bm{\varphi}_0^{-1}$, so $\boldsymbol{\mathcal R}$ describes normalized differential propagation of the initial bus synchronization error. Define its determinant as
\begin{equation}
    \Delta_{\mathcal R}(s)=\det\boldsymbol{\mathcal R}(s).
\end{equation}

\begin{theorem}[Global Synchronization Residual Integral]\label{theo: global residual integral}
Suppose $\bm{\varphi}(t)$ is element-wise in $L^1\cap C^2$ with $\dot{\bm{\varphi}},\ddot{\bm{\varphi}}\in L^1$, and $\Delta_{\mathcal R}$ satisfies the pole-zero condition in Assumption~\ref{assump: residual analyticity}. Define
\begin{equation}\label{equ: global residual quantities}
\begin{aligned}
\Gamma_{\boldsymbol{\mathcal R}}&=\sum_{\Re(z_i)>0}\Re(z_i)-\sum_{\Re(p_i)>0}\Re(p_i),\\
\gamma_0&=\Re\operatorname{tr}\!\left[\dot{\bm{\varphi}}(0^+)\bm{\varphi}_0^{-1}\right]
=\Re\operatorname{tr}\!\left\{[\dot C_0+C_0A_0]B_0(C_0B_0)^{-1}\right\}\\
&=\left.\frac{{\rm d}}{{\rm d}t}\log|\det\bm{\varphi}(t)|\right|_{0^+},
\end{aligned}
\end{equation}
where $z_i$ and $p_i$ are the RHP zeros and poles of $\Delta_{\mathcal R}$, and the subscript $0$ denotes evaluation at $t=0^+$. Then
\begin{equation}\label{equ: global bilateral integral}
\begin{aligned}
\boxed{\sum_{i=1}^n\operatorname{PV}\int_{-\infty}^{\infty}\log|\lambda_i(\boldsymbol{\mathcal R}(j\omega))|{\rm d}\omega
=\pi\gamma_0+2\pi\Gamma_{\boldsymbol{\mathcal R}}.}
\end{aligned}
\end{equation}
\end{theorem}

\begin{proof}
The proof is provided in Appendix~\ref{app:proof:global-integral}.
\end{proof}

For full reduced-state observation, $C=I$ and $B_0=I$, so $\bm{\varphi}=\Psi$. Liouville's identity, $\det\Psi(t)=\exp\{\int_0^t\operatorname{tr}[A(\tau)]{\rm d}\tau\}$, gives $\gamma_0=\Re\operatorname{tr}(A_0)$. Thus, \eqref{equ: global residual quantities} is the bus complex-angle projection of nonlinear phase-space contraction. Physically, it reflects the total available damping suggested in \cite{Peng_Impact_2025}.

For LTI systems, the above differential construction further reduces to the pole-zero structure of the observed synchronization channel.
\begin{corollary}[Global Integral for LTI Systems]\label{cor: global lti residual integral}
Consider $\dot x=Ax$ and $\boldsymbol{\varphi}=Cx$, where $A\in\mathbb C^{n_x\times n_x}$ is Hurwitz and $C\in\mathbb C^{n\times n_x}$. For $B_0\in\mathbb C^{n_x\times n}$ with $\det(CB_0)\neq0$,
\begin{equation}\label{equ: global lti residual}
\mathbf\Phi(s)=C(sI-A)^{-1}B_0,\;
\boldsymbol{\mathcal R}(s)=sC(sI-A)^{-1}B_0(CB_0)^{-1}.
\end{equation}
Let $z_i$ and $p_i$ be the finite zeros and poles of $\Delta_{\mathcal R}$. Then
\begin{equation}\label{equ: global lti pole-zero integral}
\operatorname{PV}\int_{-\infty}^{\infty}\log|\det\boldsymbol{\mathcal R}(j\omega)|{\rm d}\omega
=\pi\left(\sum_i|\Re(z_i)|-\sum_i|\Re(p_i)|\right).
\end{equation}
Furthermore, if $n=n_x$ and $C$ is nonsingular, then
$\boldsymbol{\mathcal R}(s)=sC(sI-A)^{-1}C^{-1}$ and
\begin{equation}\label{equ: global lti full-state integral}
\operatorname{PV}\int_{-\infty}^{\infty}\log|\det\boldsymbol{\mathcal R}(j\omega)|{\rm d}\omega
=\pi\Re\operatorname{tr}(A)
=\pi\sum_{i=1}^{n}\Re\lambda_i(A).
\end{equation}
\end{corollary}

\begin{proof}
The proof is provided in Appendix~\ref{app:proof:global-lti}.
\end{proof}

Equations \eqref{equ: nonlinear variational jacobians}-\eqref{equ: global bilateral integral} show that higher-order device and network dynamics enter the system-level budget through the DAE variational flow and bus complex-angle observation. In the LTI limit, this becomes the pole-zero structure in \eqref{equ: global lti pole-zero integral}. With full observation such as the pure-dVOC stations \cite{He_Passivity_2024}, it further reduces to $-\Re[\operatorname{tr}(A)]$, which has clear interpretation as the total available damping \cite{Peng_Impact_2025}.

\subsection{Bode-Type Sensitivity Conservation as a Special Case}\label{subsec: sensitivity conservation special case}
The global law reduces to the Bode-type sensitivity conservation in \cite{Peng_Stability_2026} under the following conditions. Linearize the post-disturbance dynamics and use the LTI feedback interconnection in \cite{Peng_Stability_2026}.
Consider the synchronization dynamics under additive angle--voltage disturbance $\bm d=[\Delta\bm\theta_d^T,\Delta\bm V_d^T]^T$ as
\begin{equation}\label{equ: previous sensitivity definition}
\begin{bmatrix}\Delta\bm\theta\\ \Delta\bm V\end{bmatrix}=S(s)\begin{bmatrix}\Delta\bm\theta_d\\ \Delta\bm V_d\end{bmatrix},\quad S(s)=[I+L(s)]^{-1}
\end{equation}
where $L$ is the equivalent open-loop transfer function and $S$ the sensitivity function. Assume that $L$ is rational, proper, and stable, with entry-wise relative degree greater than one; the closed loop is stable; and no RHP pole-zero cancellation occurs. Also assume $V_i^*\approx1$ p.u. for each bus, so
\begin{equation}\label{equ: log voltage small signal equivalence}
\Delta(\log V_i)=\log(V_i^*+\Delta V_i)-\log V_i^*
=\frac{\Delta V_i}{V_i^*}+o(\Delta V_i)\approx\Delta V_i.
\end{equation}
Hence, the differential complex-angle coordinates $[\Delta\bm\theta^T,\Delta\log\bm V^T]^T$ coincide with those in \eqref{equ: previous sensitivity definition}. The determinant argument in Theorem~\ref{theo: global residual integral} therefore applies to this $2n$-dimensional system.

To connect the forced sensitivity response with the initial-condition recovery, apply a constant disturbance $\bm d(t)=\bm d_0\mathbf 1(t)$ and measure the response relative to $S(0)\bm d_0$. Considering all disturbance directions gives
\begin{equation}\label{equ: sensitivity centered response}
\mathbf\Phi_S(s)=\frac{S(s)-S(0)}{s},\quad
\bm\varphi_{S,0}=S(\infty)-S(0)=I-S(0),
\end{equation}
where $S(\infty)=I$ follows from the relative-degree condition. If $I-S(0)$ is nonsingular, the corresponding residual is $\boldsymbol{\mathcal R}_S(s)=[S(s)-S(0)][I-S(0)]^{-1}$. For zero steady-state sensitivity, $S(0)=0$, this becomes
\begin{equation}\label{equ: residual equals sensitivity}
\boldsymbol{\mathcal R}_S(s)=S(s).
\end{equation}

The relative-degree condition gives $L(s)=o(|s|^{-1})$ and $S(s)=I+o(|s|^{-1})$. Hence the first-order high-frequency term vanishes and $\gamma_0=0$. The stability together with no RHP cancellation imply that $\det S$ has no RHP zeros or poles, so $\Gamma_{\boldsymbol{\mathcal R}}=0$. Substituting \eqref{equ: residual equals sensitivity} into \eqref{equ: global bilateral integral} gives $\operatorname{PV}\int_{-\infty}^{\infty}\log|\det S(j\omega)|{\rm d}\omega=0$.
For a real LTI model, $S(-j\omega)=\overline{S(j\omega)}$ and $|\det S|=\prod_{i=1}^{2n}|\lambda_i(S)|$. Thus, it becomes
\begin{equation}\label{equ: recovered previous sensitivity conservation}
\boxed{\sum_{i=1}^{2n}\int_0^\infty\log|\lambda_i(S(j\omega))|{\rm d}\omega=0,}
\end{equation}
which is precisely the sensitivity-conservation identity in \cite{Peng_Stability_2026}.

This reduction clarifies the relation between the two frameworks. The prescribed disturbance channel in \cite{Peng_Stability_2026} is a special mechanism for generating the synchronization dynamics considered here. Its recovery budget is zero, giving the classical waterbed effect. The angle--voltage sensitivity-breaking law in \cite{Peng_Stability_2026} is therefore a decomposition of this zero total budget between angle- and voltage-related eigenvalue branches.

\section{Numerical Validation}\label{sec: ieee39 validation}

The proposed mechanism is tested on an IEEE 39-bus system with synchronous generators (SGs) and heterogeneous inverters, including virtual SGs (VSGs), droop controls, and quadratic droops. The default operating point is small-signal stable, and the system and simulation settings follow \cite{Peng_Impact_2025}.
The first two experiments are conducted on the linearized model while the rest are based on nonlinear DAE dynamics.

\subsection{Frequency-Domain Conservation Across Initial Conditions and Observation Buses}\label{subsec: theorem1 universality validation}
We first verify the conservation identity across different initial-state directions and observation buses. Theorem~\ref{theo: residual integral} and Prop.~\ref{prop: lti initial condition representation} predict $\mathcal J_i^{\rm freq}=\mathcal J_i^{\rm zp}$. 
\begin{equation}\label{equ: ieee39 frequency budget} \begin{aligned} \mathcal J_i^{\rm freq}(x_0)=\int_0^\infty \log M_i(\omega;x_0){\rm d}\omega,\;
\mathcal J_i^{\rm zp}(x_0)=\frac{\pi}{2}\gamma_0+\pi\Gamma_{\mathcal R}\end{aligned} \end{equation}
Fix bus 31 as the observation bus and vary the disturbance-induced initial-state direction. The three cases are pure phase-angle $\Delta\theta_{31}$, pure voltage $\Delta\log V_{31}$, and a $45^\circ$ coupled direction in the $(\Delta\log V_{31},\Delta\theta_{31})$ plane. All are normalized to the same initial complex-angle-error magnitude.

\begin{figure}[htpb]
\centering
\includegraphics[width=0.95\linewidth]{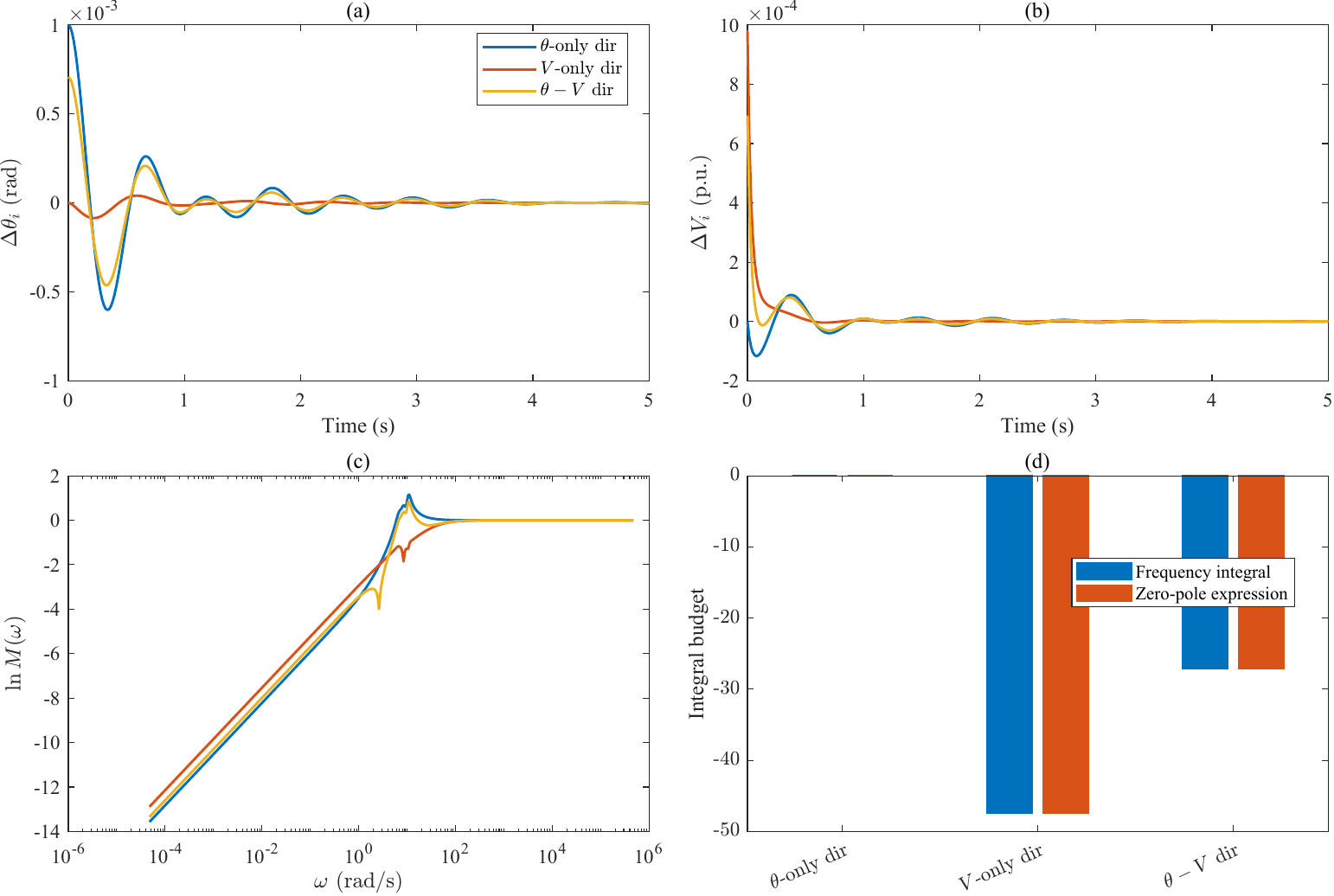}
\caption{Validation of Theorem~\ref{theo: residual integral} for different initial-state directions observed at bus 31. (a) Phase-angle deviations. (b) Voltage deviations. (c) $\log M_{31}(\omega)$. (d) Frequency integral $\mathcal J_{31}^{\rm freq}$ versus zero-pole expression $\mathcal J_{31}^{\rm zp}$.}
\label{fig: theorem1 direction validation}
\end{figure}

Fig.~\ref{fig: theorem1 direction validation}(a) and (b) show distinct phase-angle and voltage transients for the three directions. Their residual peaks also differ in location and magnitude in Fig.~\ref{fig: theorem1 direction validation}(c). Yet Fig.~\ref{fig: theorem1 direction validation}(d) gives $\mathcal J_i^{\rm freq}=\mathcal J_i^{\rm zp}$ in every case, confirming that the identity is not tied to one initial-state direction.

\begin{figure}[htpb]
\centering
\includegraphics[width=0.95\linewidth]{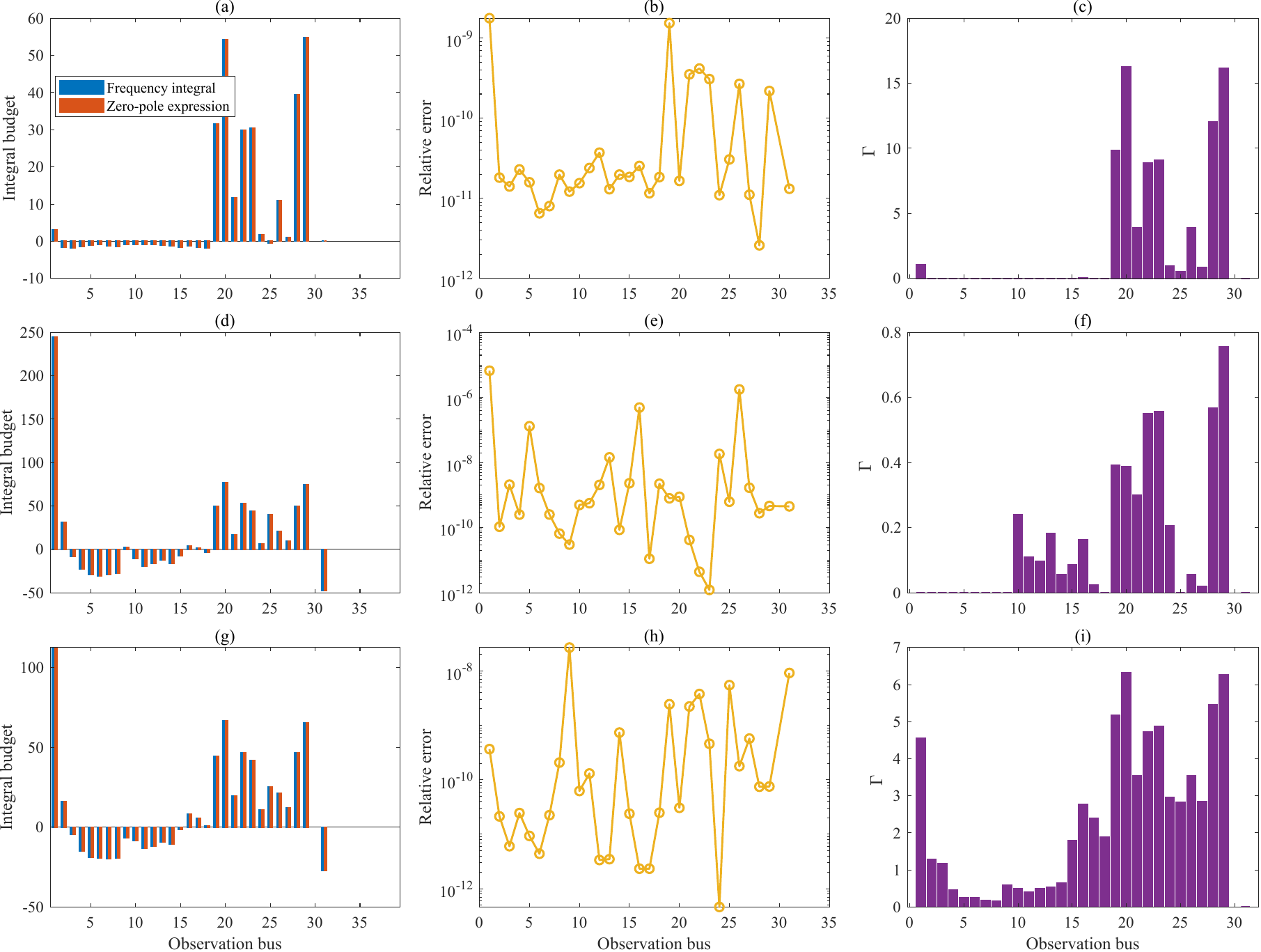}
\caption{Validation of Theorem~\ref{theo: residual integral} across observation buses. Rows correspond to the phase-angle-only, voltage-only, and $45^\circ$ coupled directions. Columns show $\mathcal J_i^{\rm freq}$ versus $\mathcal J_i^{\rm zp}$, their relative discrepancy, and the RHP-zero burden $\Gamma_{z,i}$. Missing entries indicate $C_i x_0=0$, for which the residual is undefined.}
\label{fig: theorem1 bus validation}
\end{figure}

We next fix each disturbance direction and vary the observation bus. Fig.~\ref{fig: theorem1 bus validation} shows wide variation in the integral budget across buses and directions. 
The right column shows that nonzero RHP-zero burdens are common in this inverter-dominated system. Despite these channel-dependent differences, the frequency integral and zero-pole expression agree in all valid cases. Thus, every admissible residual for a specified system and initial state obeys the same structural identity in Theorem~\ref{theo: residual integral}, with its own initial contraction rate and zero burden.

\subsection{Time-Domain Performance Limitation}\label{subsec: theorem3 validation}
We next verify the time-domain limit in Theorem~\ref{theo: frequency activity lower bound}. To isolate the observation-channel effect, fix the $45^\circ$ angle-voltage disturbance direction at bus 31 and observe complex angles at different buses.

\begin{figure}[htpb]
\centering
\includegraphics[width=1.0\linewidth]{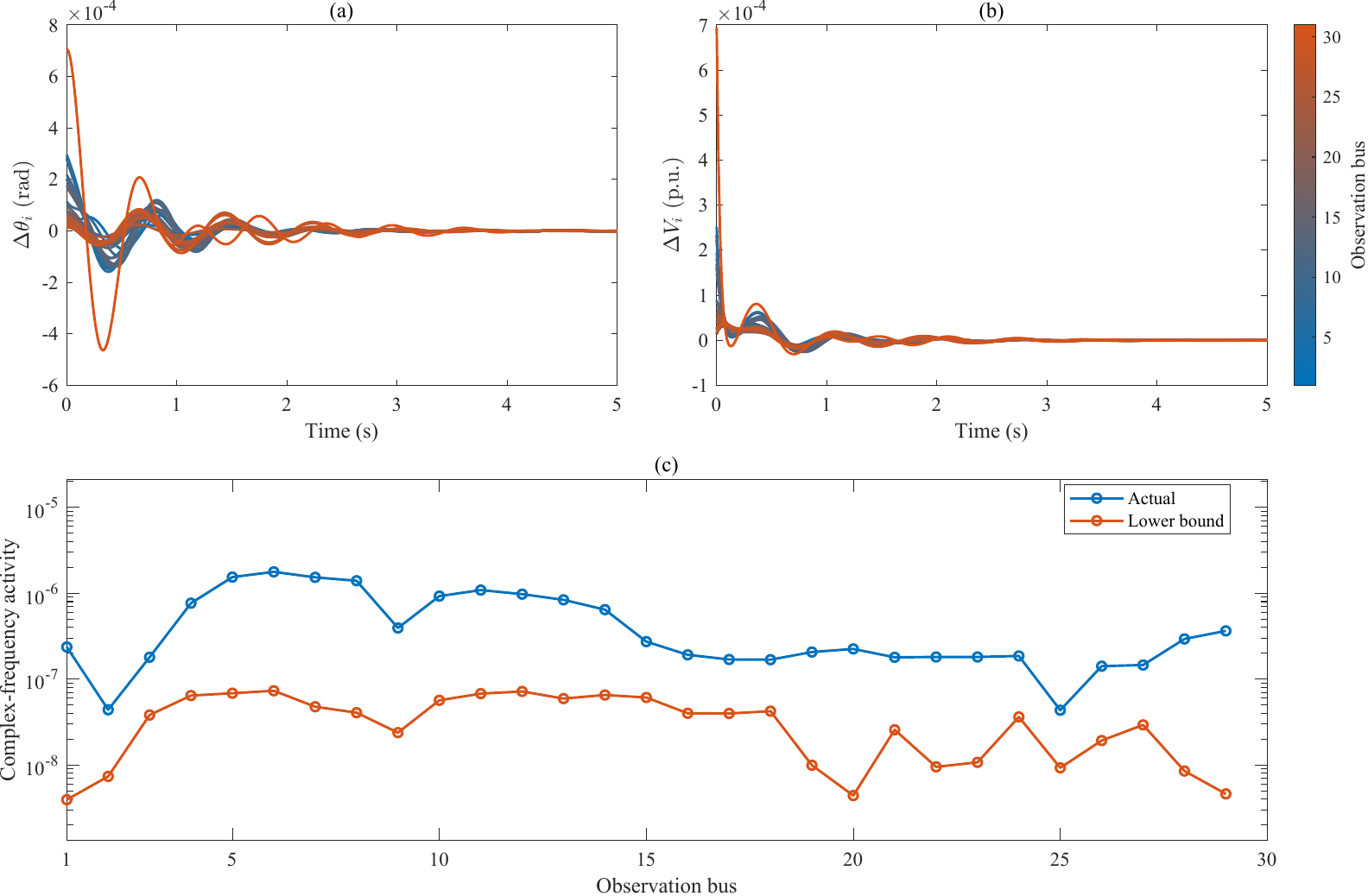}
\caption{Validation of the time-domain limit in Theorem~\ref{theo: frequency activity lower bound}. (a) Phase-angle deviations. (b) Voltage deviations; the color bar in (a)-(b) gives the observation bus. (c) Actual values and lower bounds in \eqref{equ: general frequency activity lower bound}.}
\label{fig: theorem3 bus validation}
\end{figure}
Fig.~\ref{fig: theorem3 bus validation}(a) and (b) show different phase-angle and voltage transients across buses. This heterogeneity also appears in the actual activity and RHP-zero lower bound. Still, Fig.~\ref{fig: theorem3 bus validation}(c) shows $\int_0^\infty|\varpi_i(t)|^2{\rm d}t\geq2|\varphi_i(0)|^2\Gamma_{\mathcal R_i}$ for every valid bus. Thus, an RHP residual zero imposes unavoidable combined voltage-rate and frequency-deviation activity during error removal. The bound identifies a structural recovery cost rather than the full transient waveform.

\subsection{Performance Redistribution Through Damping Tuning}\label{subsec: damping redistribution}
The method has a direct control implication: if tuning preserves the initial contraction rate and RHP-zero burden, it cannot change the logarithmic residual budget. It can only redistribute performance across frequencies.

A simplified swing equation $m_\delta\ddot\delta+d_\delta\dot\delta+k_s\delta=0$ illustrates this mechanism,
where $m_\delta$, $d_\delta$, and $k_s$ are the inertia, damping, and synchronizing coefficients. Its residual is
\begin{equation}\label{equ: reduced damping residual}
\begin{aligned}
    \mathcal R_d(s)
    ={s(m_\delta s+d_\delta)}/
    {(m_\delta s^2+d_\delta s+k_s)}.
\end{aligned}
\end{equation}
All poles and the nonstructural zero $-d_\delta/m_\delta$ lie in the LHP. Also, $\gamma_0=\dot\delta(0)/\delta_0=0$, so $\mathcal J=0$. The low-frequency residual satisfies
\begin{equation}\label{equ: reduced damping low frequency}
\begin{aligned}
    \mathcal R_d(j\omega)
    =j\omega\frac{d_\delta}{k_s}+{o}(\omega),\,\omega\to0,\quad
    \frac{1}{\delta_0}\int_0^\infty\delta(t){\rm d}t
    =\frac{d_\delta}{k_s}.
\end{aligned}
\end{equation}
Thus, stronger damping can suppress the oscillatory peak while raising the low-frequency residual envelope and accumulated angle error \cite{Peng_Stability_2026}. This reduced model gives a power-system interpretation of the redistribution in Theorem~\ref{theo: residual integral}.

For verification, scale the damping coefficient of the VSG at bus 31 by $\rho_d\in[1,4]$, with other settings fixed. The representative values are $\rho_d=1$, $1.7$, $2.5$, and $4$. All selected cases are stable, satisfy $\gamma_0=0$, and have no RHP residual zero; hence $\mathcal J=0$ throughout the scan.

\begin{figure}[htpb]
\centering
\includegraphics[width=1.0\linewidth]{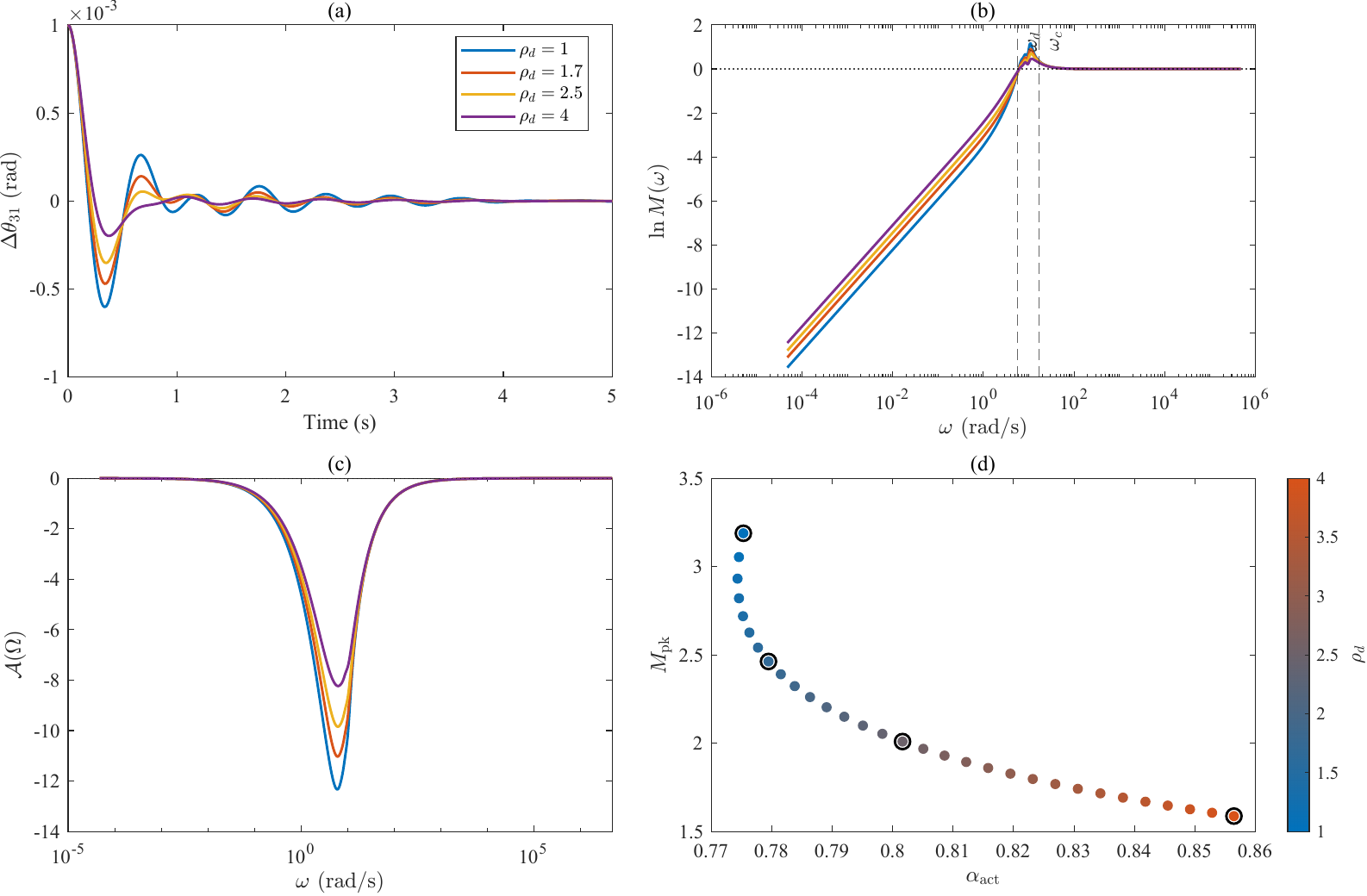}
\caption{{Synchronization recovery performance redistribution under VSG damping tuning at bus 31.} (a) Phase-angle recovery. (b) Increasing $\rho_d$ suppresses $M_{\rm pk}$ but raises the low-frequency residual. (c) Cumulative areas $\mathcal A_{\rho_d}(\Omega)=\int_0^\Omega\log M_{\rho_d}(\omega){\rm d}\omega$ converge to the same zero budget. (d) Trade-off between $\alpha_{\rm act}$ and $M_{\rm pk}$; circles mark the cases in (a)-(c).}
\label{fig: damping redistribution}
\end{figure}

Fig.~\ref{fig: damping redistribution}(a) shows that increasing $\rho_d$ reduces phase-angle oscillations. Fig.~\ref{fig: damping redistribution}(b), however, shows that the improvement is not uniform: the intermediate-frequency peak falls while the low-frequency part of $\log M(\omega)$ rises, consistent with \eqref{equ: reduced damping low frequency}. In Fig.~\ref{fig: damping redistribution}(c), the cumulative areas differ in the intermediate range but all return to about zero at high frequency. Damping tuning therefore shifts logarithmic area from oscillatory to low-frequency recovery rather than removing the constraint.

To quantify this trade-off, define the following quantities
\begin{equation}\label{equ: actual low frequency envelope}
    \alpha_{\rm act}(\rho_d)
    =\sup_{0<\omega\leq\omega_d}
    \frac{\omega_d}{\omega}M_{\rho_d}(\omega),
    \;
    M_{\rm pk}(\rho_d)
    =\sup_{\omega_d\leq \omega\leq\omega_c}M_{\rho_d}(\omega).
\end{equation}
Fig.~\ref{fig: damping redistribution}(d) shows that reducing $M_{\rm pk}$ increases $\alpha_{\rm act}$. Stronger damping therefore suppresses the oscillatory peak at the cost of greater low-frequency synchronization error. {In other words, synchronization recovery performance cannot improve uniformly while the recovery budget remains unchanged.}

To directly validate the finite-frequency bound in Theorem~\ref{theo: finite frequency peak}, we use the same four damping cases and a common frequency band $\omega_d\leq\omega\leq\omega_c$ with $k=1$.
For all four cases, $\gamma_0=0$ and no RHP residual zero exists, so $\mathcal J=0$.
The low-frequency parameter increases from $\alpha=0.4279$ to $0.6471$ as $\rho_d$ increases, while the high-frequency parameter remains in the narrow range $\delta=0.0539$--$0.0560$.
Thus, $\alpha<1$ and $\delta<1/2$ hold in every case, and the numerical checks confirm both envelope conditions without positive violation.

\begin{figure}[htpb]
\centering
\includegraphics[width=\linewidth]{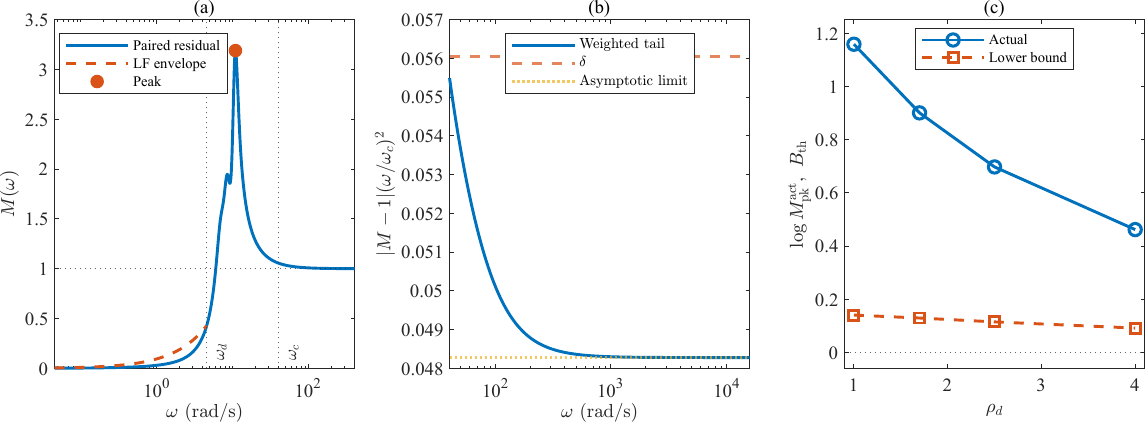}
\caption{Numerical validation of Theorem~\ref{theo: finite frequency peak} under VSG damping tuning at bus 31. (a) Representative paired residual magnitude for $\rho_d=1$, together with the low-frequency envelope and the common boundaries $\omega_d=4.5309$ rad/s and $\omega_c=40.1581$ rad/s. (b) Weighted high-frequency tail $|M(\omega)-1|(\omega/\omega_c)^2$ and its envelope. (c) Actual intermediate-frequency peaks and their theoretical lower bounds for the four damping cases.}
\label{fig: theorem2 peak bound validation}
\end{figure}

Denote $B_{\rm th}$ by the right-hand side of \eqref{equ: peak lower bound} in Theorem~\ref{theo: finite frequency peak}.
In this case, the resulting theoretical bounds are positive for all four cases. For $\rho_d=1$, $1.7$, $2.5$, and $4$, respectively, Theorem~\ref{theo: finite frequency peak} gives
$e^{B_{\rm th}}=1.1507$, $1.1380$, $1.1216$, and $1.0958$, whereas the actual peaks are
$M_{\rm pk}=3.1887$, $2.4628$, $2.0094$, and $1.5874$.
Hence, every actual peak lies above its predicted lower bound, as shown in Fig.~\ref{fig: theorem2 peak bound validation}(c).
Because $B_{\rm th}>0$ in all cases, the theorem predicts $M_{\rm pk}>1$ before the intermediate-frequency peak is evaluated.
Increasing damping lowers both the actual peak and its unavoidable lower bound, while increasing the low-frequency envelope.
This quantitatively confirms the peak--bandwidth trade-off and complements the redistribution mechanism observed in Fig.~\ref{fig: damping redistribution}.

\subsection{Validation with the Nonlinear DAE Model}
\label{subsec: nonlinear trajectory validation}

We further test the complete nonlinear DAE model under a short-circuit fault at bus~16. The fault is applied at $t=-0.10$~s and cleared at $t=0$~s. Fig.~\ref{fig: exp7 nonlinear trajectory}(a) shows the complex-angle error, while Fig.~\ref{fig: exp7 nonlinear trajectory}(b) shows its contraction toward the origin.

\begin{figure}[htpb]
    \centering
    \includegraphics[width=\linewidth]
    {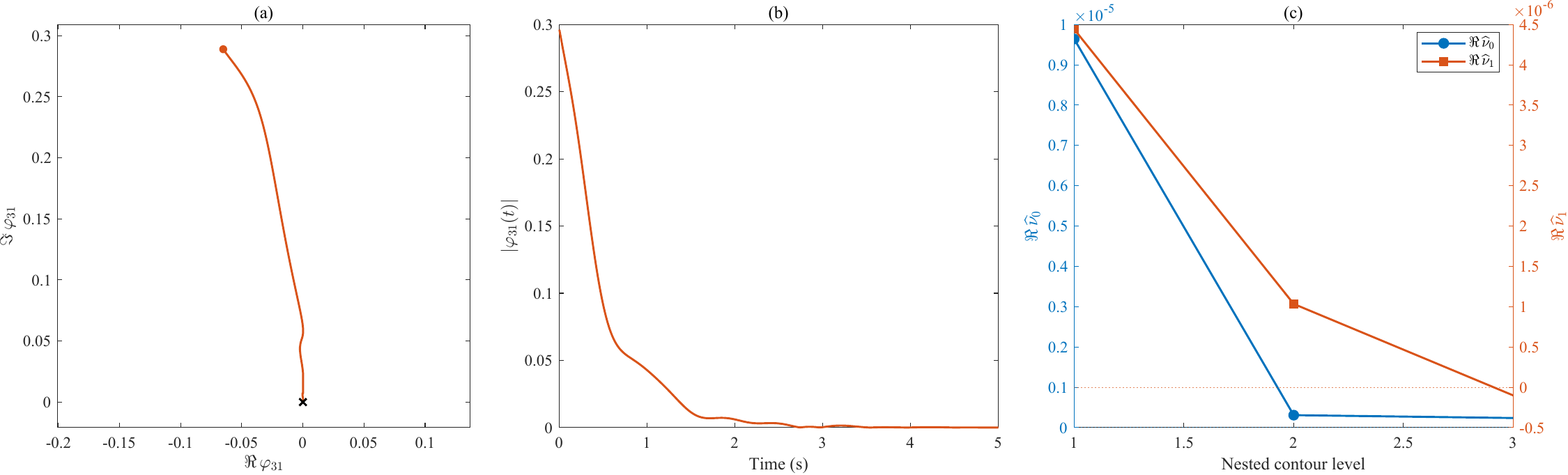}
    \caption{Validation using the nonlinear post-fault response of the 31st bus under a short-circuit fault.
    (a) Recovery trajectory of $\varphi_1(t)$ in the complex-angle plane.
    (b) Time-domain decay of $|\varphi_1(t)|$.
    (c) Zeroth- and first-order moments over three progressively enlarged RHP contours.}
    \label{fig: exp7 nonlinear trajectory}
\end{figure}

Following Section~\ref{subsec: trajectory based computation}, use three nested regions $\mathcal D_{\epsilon_r,S_r,\Omega_r}$ in \eqref{equ: D_r search region}: $(\epsilon_r,S_r,\Omega_r)=(10^{-3},2,10)$, $(5\times10^{-4},4,20)$, and $(2.5\times10^{-4},8,40)$. In Fig.~\ref{fig: exp7 nonlinear trajectory}(c), neither the first nor second domain encloses an RHP zero. The moments remain unchanged after expansion to the third domain, providing a convergence check. At the accepted third level, $\widehat{\nu}_0=2.373141757\times10^{-7}-j1.630704867\times10^{-8}$ and $\widehat{\nu}_1=-1.010914073\times10^{-7}-j1.851216607\times10^{-7}$. Thus, $\widehat{\nu}_0\approx0$ indicates no enclosed RHP zero and $\widehat{\Gamma}_{\mathcal R}=\Re(\widehat{\nu}_1)=-1.010914\times10^{-7}\approx0$. Combined with $\gamma_0$, this gives $\mathcal J^{\rm freq}=-1.929765$ and $\mathcal J^{\rm zp}=-1.918402$, with a $0.592\%$ normalized closure error from truncation.

\section{Conclusion}\label{sec: conclusion}

This paper establishes a conservation-based view of power-system synchronization recovery in complex frequency representation. The main result shows that the recovery performance is constrained by a logarithmic budget determined by the initial contraction and residual zero--pole structure, which are physically determined by system structures and initial states. 
Consequently, controller tuning generally redistributes synchronization degradation rather than removing it, leading to unavoidable frequency-domain and time-domain tradeoffs. The system-level extension further connects this result with classical Bode-type sensitivity conservation, which appears as a zero-budget special case. Studies on the heterogeneous IEEE 39-bus system confirm these mechanisms for different initial conditions, observation buses, damping settings, and nonlinear post-fault responses. These results suggest that synchronization control should focus on redistributing the constrained recovery cost toward less critical frequencies or channels, or on changing the underlying structural budget itself.

\section*{Acknowledgment}
The authors used OpenAI ChatGPT to assist with the coding of numerical validation and with language editing throughout the manuscript. All AI-assisted code and text were reviewed and verified by the authors.

\bibliographystyle{IEEEtran}
\bibliography{references}

\clearpage
\appendices
\section{Proofs of Main Results}\label{appendix: proofs}
For completeness, the proofs of the results stated in the main text are collected below.

\subsection{Proof of Proposition~\ref{prop: residual frequency representation}}\label{app:proof:residual-correction}

\begin{proof}
Since $\dot\varphi=\varpi$, integration by parts in \eqref{equ: Phi definition} gives $s\Phi(s)=\varphi(0)+\Omega(s)$. Division by $\varphi(0)$ proves \eqref{equ: correction function}. Equation \eqref{equ: residual definition} gives $\mathcal R(s)/s=\Phi(s)/\varphi(0)$, whose Laplace transform yields \eqref{equ: residual and correction step response}. Since $\mathcal C=1-\mathcal R$, ${\mathcal C(s)}/{s}={1}/{s}-{\Phi(s)}/{\varphi(0)}$, as in \eqref{equ: correction function}. Finally, $\varpi\in L^1$ and $\varphi(t)\to0$ imply $\Omega(0)=\int_0^\infty\varpi(t){\rm d}t
    =\varphi(\infty)-\varphi(0)=-\varphi(0)$.
Thus, \eqref{equ: correction function} gives $\mathcal R(0)=0$ and $\mathcal C(0)=1$. The Laplace transform of an $L^1$ function vanishes as $|s|\to\infty$ in the closed RHP, so $\Omega(s)\to0$, $\mathcal R(s)\to1$, and $\mathcal C(s)\to0$.
\end{proof}

\subsection{Proof of Theorem~\ref{theo: residual integral}}\label{app:proof:residual-integral}

\begin{proof}
For $a\in\mathbb C_+$, define the Blaschke factor $B_a(s)={(s-a)}/{(s+\overline a)}$ \cite{Seron_Fundamental_1997}.
For $s=j\omega$, the numerator and denominator in the Blaschke factor have the same magnitude, and therefore $|B_a(j\omega)|=1$. Moreover,
\begin{equation}\label{equ: blaschke asymptote}
\begin{aligned}
    B_a(s)
    =\left(1-\frac{a}{s}\right)
      \left(1+\frac{\overline a}{s}\right)^{-1}
    =1-\frac{2\Re(a)}{s}+o(|s|^{-1}).
\end{aligned}
\end{equation}

Construct
\begin{equation}\label{equ: zero free H}
    H(s)=\mathcal R(s)
    \frac{\prod_i B_{p_i}(s)}{\prod_i B_{z_i}(s)}.
\end{equation}
By construction, $H$ is analytic and zero-free in $\mathbb C_+$, and the unit-modulus property gives $|H(j\omega)|=|\mathcal R(j\omega)|$.

Using \eqref{equ: analytic asymptote} and \eqref{equ: blaschke asymptote}, one obtains
\begin{equation}\label{equ: H asymptote}
    H(s)=1+\frac{h_1}{s}+o(|s|^{-1}),
    \qquad h_1=a_1+2\Gamma_{\mathcal R}.
\end{equation}
Since $\mathbb C_+$ is simply connected and $H$ is zero-free, $H$ admits an analytic logarithm in $\mathbb C_+$. Choose the branch satisfying $\log H(s)\rightarrow0$ as $s\rightarrow\infty$. Equation \eqref{equ: H asymptote} then yields
\begin{equation}\label{equ: log H asymptote}
    \log H(s)=\frac{h_1}{s}+o(|s|^{-1}).
\end{equation}

Consider the clockwise boundary of the right half-disk of radius $R$, with a RHP indentation of radius $\epsilon$ around the structural zero at the origin. Near the origin, $H(s)=s^m\widetilde H(s)$ for some integer $m\geq1$ and analytic $\widetilde H(0)\neq0$. Therefore, there is a constant $C>0$ such that $|\log H(s)|\leq C(1+|\log\epsilon|)$ on the indentation. Since its length is $\pi\epsilon$, the magnitude of its contribution is bounded by $\pi C\epsilon(1+|\log\epsilon|)$, which tends to zero as $\epsilon\to0$.

Now, Cauchy's theorem is applicable to $\log H$ on the indented domain, i.e., $\oint \log H(\omega){\rm d}\omega=0$. Along the imaginary-axis segments, $s=j\omega$ and ${\rm d}s=j{\rm d}\omega$, giving
\begin{equation}
    j\left(\int_{-R}^{-\epsilon}+\int_{\epsilon}^{R}\right)
    \log H(j\omega){\rm d}\omega.
\end{equation}
For the $R$-radius large semicircle, we parameterize it clockwise by $s=R\exp(j\phi)$, where $\phi$ decreases from $\pi/2$ to $-\pi/2$. Then, \eqref{equ: log H asymptote} yields
\begin{equation}
\begin{aligned}
    \int_{C_R}\log H(s){\rm d}s
    &=\int_{\pi/2}^{-\pi/2}
      \left(\frac{h_1}{Re^{j\phi}}+o(R^{-1})\right)
      jRe^{j\phi}{\rm d}\phi\\
    &=-j\pi h_1+o(1).
\end{aligned}
\end{equation}
Letting $\epsilon\rightarrow0$ and then $R\rightarrow\infty$ in the contour identity gives
\begin{equation}\label{equ: complex log integral}
    \operatorname{PV}\int_{-\infty}^{\infty}
    \log H(j\omega){\rm d}\omega=\pi h_1.
\end{equation}
Taking real parts in \eqref{equ: complex log integral}, using $|H(j\omega)|=|\mathcal R(j\omega)|$, and substituting $h_1=a_1+2\Gamma_{\mathcal R}$ and \eqref{equ: residual budget definition} prove \eqref{equ: bilateral integral theorem}.

Moreover,
\begin{equation}
    \frac{{\rm d}}{{\rm d}t}\log|\varphi(t)|
    =\frac{\Re\{\overline{\varphi(t)}\varpi(t)\}}{|\varphi(t)|^2}
    =\Re\left(\frac{\varpi(t)}{\varphi(t)}\right).
\end{equation}
Combining it with \eqref{equ: a1}, its value at $t=0^+$ is $\Re(a_1)=\gamma_0$. Substitution into \eqref{equ: residual budget definition} proves \eqref{equ: physical integral theorem}.
\end{proof}

\subsection{Proof of Corollary~\ref{cor: stable integral}}\label{app:proof:stable-response}

\begin{proof}
The Laplace transform of an $L^1$ signal is analytic in $\mathbb C_+$. Multiplication by $s/\varphi(0)$ introduces no RHP pole, so the pole sum in \eqref{equ: gamma zp definition} is empty. \eqref{equ: stable residual integral} then follows from \eqref{equ: residual budget definition}.
\end{proof}

\subsection{Proof of Proposition~\ref{prop: zero cancellation}}\label{app:proof:weighted-cancellation}

\begin{proof}
For $z\in\mathbb C_+$, \eqref{equ: Phi definition} converges, and thus \eqref{equ: residual definition} gives $\mathcal R(z)=z\Phi(z)/\varphi(0)$. Since $z\neq0$, $\mathcal R(z)=0$ if and only if $\Phi(z)=0$. Substituting $z=\sigma+j\omega_z$ into \eqref{equ: Phi definition} completes the proof.
\end{proof}

\subsection{Proof of Proposition~\ref{prop: lti initial condition representation}}\label{app:proof:lti-initial-condition}

\begin{proof}
The system dynamics give $x(t)=\exp(At)x_0$ and hence $\Phi_i(s)$ in \eqref{equ: residual lti realization}. Using $s(sI-A)^{-1}=I+A(sI-A)^{-1}$ and $\varphi_i(0)=C_ix_0$ in \eqref{equ: residual definition} gives $\mathcal R_i(s)$ in \eqref{equ: residual lti realization}. For $z\neq0$, $\mathcal R_i(z)=0$ if and only if $\Phi_i(z)=0$. The pole statement follows after removing uncontrollable, unobservable, and exactly cancelled modes.
\end{proof}

\subsection{Proof of Proposition~\ref{cor: conventional transfer zeros}}\label{app:proof:transfer-zeros}

\begin{proof}
An impulse through $B_dq$ gives $x(0^+)=B_dq$ and $\varphi_i(t)=C_i\exp(At)B_dq$. Its Laplace transform is $G_{i,q}(s)=C_i(sI-A)^{-1}B_dq$, proving \eqref{equ: residual from physical channel}. The nonzero scale factor $1/(C_iB_dq)$ does not change zeros or poles, while multiplication by $s$ adds the structural zero at the origin. Exact cancellations are removed by minimal realization.
\end{proof}

\subsection{Proof of Theorem~\ref{theo: finite frequency peak}}\label{app:proof:finite-frequency-peak}

\begin{proof}
For $|w|<1/2$, $|\log(1+w)|\leq3|w|/2$ \cite{Peng_Stability_2026}. Substituting this fact yields
\begin{equation}
\begin{aligned}
|\log M(\omega)|\leq\frac32|M(\omega)-1|\leq\frac{3\delta}{2}
\left(\frac{\omega_c}{\omega}\right)^{1+k}
\end{aligned}
\end{equation}
for $\omega\geq\omega_c$.
Therefore,
\begin{equation}\label{equ: tail area bound}
    \left|\int_{\omega_c}^{\infty}
    \log M(\omega){\rm d}\omega\right|
    \leq\frac{3\delta\omega_c}{2k}=\beta_c.
\end{equation}

Split \eqref{equ: positive frequency budget} over $[0,\omega_d]$, $(\omega_d,\omega_c)$, and $[\omega_c,\infty)$. The low-frequency condition, $M_{\rm pk}$, and \eqref{equ: tail area bound} give
\begin{equation}\label{equ: freq-wise integral of logM}
    \begin{aligned}
        &\int_0^{\omega_d}\log M(\omega){\rm d}\omega
    \leq\omega_d(\log\alpha-1),\, \int_{\omega_c}^{\infty}\log M(\omega){\rm d}\omega
    \leq\beta_c.\\
    &\int_{\omega_d}^{\omega_c}\log M(\omega){\rm d}\omega
    \leq(\omega_c-\omega_d)\log M_{\rm pk},\\
    \end{aligned}
\end{equation}
Adding these inequalities and using \eqref{equ: positive frequency budget} gives $\mathcal J\leq\omega_d(\log\alpha-1)
    +(\omega_c-\omega_d)\log M_{\rm pk}+\beta_c$, proving \eqref{equ: peak lower bound}.
\end{proof}

\subsection{Proof of Corollary~\ref{cor: bandwidth lower bound}}\label{app:proof:bandwidth}

\begin{proof}
Under the prescribed bounds,
\begin{equation}
\mathcal J
\leq\kappa\omega_c(\log\alpha-1)
 +(1-\kappa)\omega_c\log\overline M
 +\frac{3\delta\omega_c}{2k}=\omega_c d_{\rm bw}.
\end{equation}
This proves \eqref{equ: bandwidth budget condition} when $d_{\rm bw}>0$. If $d_{\rm bw}<0$, then $\omega_c d_{\rm bw}<0$ for every $\omega_c>0$, contradicting $\mathcal J>0$.
\end{proof}

\subsection{Proof of Theorem~\ref{theo: frequency activity lower bound}}\label{app:proof:activity}

\begin{proof}
For brevity, define $G(s)=\mathcal R(s)-1=\int_0^\infty \exp({-st})g(t){\rm d}t$.
For every complex number $u$ such that $1+u\neq0$,
\begin{equation}\label{equ: logarithm quadratic inequality}
\log|1+u|=\frac12\log(1+2\Re(u)+|u|^2)
\leq\Re(u)+\frac12|u|^2,
\end{equation}
Applying \eqref{equ: logarithm quadratic inequality} to $u=G(j\omega)$ and integrating over the entire imaginary axis gives
\begin{equation}\label{equ: integrated logarithm inequality}
\begin{aligned}
&\operatorname{PV}\int_{-\infty}^{\infty}
\log|\mathcal R(j\omega)|{\rm d}\omega\\
&\quad\leq
\operatorname{PV}\int_{-\infty}^{\infty}\Re G(j\omega){\rm d}\omega
+\frac12\int_{-\infty}^{\infty}|G(j\omega)|^2{\rm d}\omega.
\end{aligned}
\end{equation}

Since $g$ is continuous and $\dot g\in L^1$, the one-sided Fourier inversion formula gives
\begin{equation}\label{equ: initial value inversion}
    \operatorname{PV}\int_{-\infty}^{\infty}G(j\omega){\rm d}\omega
    =\pi g(0^+).
\end{equation}
Since $g(0^+)=\varpi(0^+)/\varphi(0)$, the real part of \eqref{equ: initial value inversion} is $\pi\gamma_0$.
Substituting \eqref{equ: bilateral integral theorem}, \eqref{equ: residual budget definition}, and \eqref{equ: initial value inversion} into \eqref{equ: integrated logarithm inequality} gives
\begin{equation}
    \int_{-\infty}^{\infty}|G(j\omega)|^2{\rm d}\omega
    \geq4\pi\Gamma_{\mathcal R}.
\end{equation}
If $\Gamma_{\mathcal R}<0$, the non-negative left-hand side indicates the positive-part form. Parseval's identity yields
\begin{equation}
\begin{aligned}
    \int_{-\infty}^{\infty}|G(j\omega)|^2{\rm d}\omega
    =2\pi\int_0^\infty|g(t)|^2{\rm d}t=\frac{2\pi}{|\varphi(0)|^2}
      \int_0^\infty|\varpi(t)|^2{\rm d}t,
\end{aligned}
\end{equation}
Combining both equations completes the proof.
\end{proof}

\subsection{Proof of Theorem~\ref{theo: global residual integral}}\label{app:proof:global-integral}

\begin{proof}
Equations \eqref{equ: nonlinear variational flow}-\eqref{equ: nonlinear differential synchronization error} give $\dot{\bm{\varphi}}(0^+)=(\dot C_0+C_0A_0)B_0$. Entry-wise integration by parts in \eqref{equ: nonlinear matrix residual definition} yields
$\boldsymbol{\mathcal R}(s)=I+\dot{\bm{\varphi}}(0^+)\bm{\varphi}_0^{-1}/s+o(|s|^{-1})$.
Therefore, $\Delta_{\mathcal R}(s)=1+\operatorname{tr}[\dot{\bm{\varphi}}(0^+)\bm{\varphi}_0^{-1}]/s+o(|s|^{-1})$. Applying Thm.~\ref{theo: residual integral} to $\Delta_{\mathcal R}$ and noticing $\log|\det({\mathcal R})|=\sum_i\log|\lambda_i({\mathcal R})|$ prove \eqref{equ: global bilateral integral}.
\end{proof}

\subsection{Proof of Corollary~\ref{cor: global lti residual integral}}\label{app:proof:global-lti}

\begin{proof}
For the LTI model, $C(t)=C$ and $\Psi(t)=\exp({At})$, so \eqref{equ: global lti residual} follows from \eqref{equ: nonlinear differential synchronization error}-\eqref{equ: nonlinear matrix residual definition}. Since $\Delta_{\mathcal R}(\infty)=1$, normalized pole-zero factorization gives $\gamma_0=\Re(\sum_i p_i-\sum_i z_i)$. Combining this relation with \eqref{equ: global bilateral integral} gives \eqref{equ: global lti pole-zero integral}. If $n=n_x$, then $B_0(CB_0)^{-1}=C^{-1}$. The only zeros are structural zeros at the origin, and the poles are $\lambda_i(A)$, giving \eqref{equ: global lti full-state integral}.
\end{proof}

\section{Numerical Evaluation of Residual Integral Based on Nonlinear Trajectories }
\label{appendix: nonlinear trajectory frequency integral}

This part details the numerical evaluation of the frequency-side residual budget in the nonlinear fault-clearing experiment conduction in Section~\ref{subsec: nonlinear trajectory validation}. For a sampled post-fault trajectory $\varphi(t)$, define
\begin{equation}\label{equ: appendix paired log residual}
\begin{aligned}
    \ell(\omega)&=\log M(\omega)
    =\frac{1}{2}\log\left|
    \mathcal R(j\omega)\mathcal R(-j\omega)
    \right|,\\
    \mathcal J^{\rm freq}&=\int_0^\infty\ell(\omega){\rm d}\omega.
\end{aligned}
\end{equation}
The positive- and negative-frequency branches are evaluated separately because the complex-valued trajectory does not generally possess conjugate symmetry.

\subsection{Direct Evaluation over the Resolved Frequency Range}

Conduct the transient simulation until $|\varphi(T)|/|\varphi(0)|$ is sufficiently small.
For the samples $\varphi_n=\varphi(t_n)$, 
the finite-horizon Laplace transform is evaluated by composite trapezoidal quadrature with sufficient time interval $\Delta t_n=t_{n+1}-t_n$ as 
\begin{equation}\label{equ: appendix finite trajectory laplace}
\widehat\Phi_T(s)
=\sum_{n=0}^{N-1}\frac{\Delta t_n}{2}
\left[
    e^{-st_n}\varphi_n
    +e^{-st_{n+1}}\varphi_{n+1}
\right].
\end{equation}
The corresponding paired residual samples are
\begin{equation}\label{equ: appendix sampled paired residual}
\widehat{\mathcal R}(\pm j\omega)
=\frac{\pm j\omega\widehat\Phi_T(\pm j\omega)}{\varphi(0)},
\quad
\widehat\ell(\omega)
=\frac{1}{2}\log\left|
\widehat{\mathcal R}(j\omega)
\widehat{\mathcal R}(-j\omega)
\right|.
\end{equation}

Split the frequency spectrum into three parts, $[0,\omega_{\min}]$, $(\omega_{\min},\omega_{\max}]$, and $(\omega_{\max},\infty)$.
These parts are computed in different approximations, as shown below.

In numerical verification, since the structural zero gives $M(\omega)=c_0\omega+\mathcal O(\omega)$ as $\omega\to0^+$, the interval below $\omega_{\min}$ satisfies
\begin{equation}\label{equ: appendix low frequency correction}
\int_0^{\omega_{\min}}\ell(\omega){\rm d}\omega
\approx
\omega_{\min}\bigl[\widehat\ell(\omega_{\min})-1\bigr].
\end{equation}

For the medium band $(\omega_{\min},\omega_{\max}]$, $\widehat\ell$ is evaluated on a logarithmic grid beginning at $\omega_{\min}=10^{-5}$~rad/s with 160 points per decade. The directly resolved contribution is calculated by trapezoidal integration in $\omega$. 
Consequently, the directly computed integral of the finite-frequency parts $[0, \omega_{\max}] = [0,\omega_{\min}]\cup(\omega_{\min},\omega_{\max}]$ is
\begin{equation}\label{equ: appendix direct frequency integral}
\widehat{\mathcal J}_{\rm grid}(\omega_{\max})
=
\omega_{\min}\bigl[\widehat\ell(\omega_{\min})-1\bigr]
+
\int_{\omega_{\min}}^{\omega_{\max}}
\widehat\ell(\omega){\rm d}\omega.
\end{equation}

\subsection{High-Frequency Tail Estimate}

The remaining integral above $(\omega_{\max},\infty)$ is estimated independently from the local initial behavior of the nonlinear trajectory. A fourth-order polynomial is fitted to the first $2$~ms of $\varphi(t)$, from which $\dot\varphi(0^+)$ and $\ddot\varphi(0^+)$ are obtained. Define
\begin{equation}\label{equ: appendix high frequency coefficients}
    a_1=\frac{\dot\varphi(0^+)}{\varphi(0)},
    \qquad
    a_2=\frac{\ddot\varphi(0^+)}{\varphi(0)}.
\end{equation}
The initial-value expansion of the residual is then
\begin{equation}\label{equ: appendix high frequency residual expansion}
    \mathcal R(s)
    =1+\frac{a_1}{s}+\frac{a_2}{s^2}
    +{\mathcal O}(|s|^{-3}).
\end{equation}
Pairing the two imaginary-axis branches cancels the first-order terms and yields
\begin{equation}\label{equ: appendix high frequency log expansion}
\ell(\omega)
=
\frac{c_{\rm hf}}{\omega^2}
+{\mathcal O}(\omega^{-4}),
\qquad
c_{\rm hf}
=\frac{1}{2}\Re\left(a_1^2-2a_2\right).
\end{equation}
Hence, the unresolved high-frequency contribution is approximated by
\begin{equation}\label{equ: appendix high frequency tail}
\widehat{\mathcal J}_{\rm tail}(\omega_{\max})
=
\int_{\omega_{\max}}^\infty\ell(\omega){\rm d}\omega
\approx
\frac{c_{\rm hf}}{\omega_{\max}}.
\end{equation}
This estimate uses only the initial local derivatives of the nonlinear trajectory and is not obtained by forcing agreement with the zero-pole expression in \eqref{equ: physical integral theorem}.

The upper integration limit is initialized at $80$~rad/s and enlarged by a factor of $1.5$, subject to the sampling-based safety limit
$\omega_{\max}\leq0.04\pi/\Delta t_{\max}$. At each extension, the total estimate
\begin{equation}\label{equ: appendix total nonlinear frequency budget}
\widehat{\mathcal J}^{\rm freq}
=
\widehat{\mathcal J}_{\rm grid}(\omega_{\max})
+
\widehat{\mathcal J}_{\rm tail}(\omega_{\max})
\end{equation}
is recomputed. The calculation is accepted when two successive total estimates differ by no more than
$0.03+0.03\max\{1,|\widehat{\mathcal J}^{\rm freq}|\}$ and the last $5\%$ of the resolved frequency samples satisfy
\begin{equation}\label{equ: appendix tail shape check}
\frac{\operatorname{median}
\left|\omega^2\widehat\ell(\omega)-c_{\rm hf}\right|}
{\max\{1,|c_{\rm hf}|\}}
\leq0.15.
\end{equation}
Thus, the reported nonlinear frequency budget combines a directly integrated low-to-high-frequency contribution with an independently estimated asymptotic tail. 

\end{document}